\documentclass[12pt, a4paper, twoside]{scrartcl}
\pdfoutput=1
\usepackage[disable]{todonotes}

\usepackage[tracking=true,expansion=true,stretch=15,shrink=15]{microtype}
\DeclareMicrotypeSet*[tracking]{my}
{ font = */*/*/sc/* } 
\SetTracking{encoding = *, shape = sc }{ 45 }
\SetProtrusion{encoding={*},family={bch},series={*},size={6,7}}
{1={ ,750},2={ ,500},3={ ,500},4={ ,500},5={ ,500},
	6={ ,500},7={ ,600},8={ ,500},9={ ,500},0={ ,500}}
\SetExtraKerning[unit=space]
{encoding={*}, family={qhv}, series={b}, size={large,Large}}
{1={-200,-200}, 
	\textendash={400,400}}

\usepackage{setspace}

\usepackage{url}
\usepackage{latexsym}
\usepackage{amssymb}
\usepackage{amsmath}
\usepackage{amsthm}
\usepackage{dsfont}
\usepackage{mathrsfs}
\usepackage{empheq}
\usepackage{bm}
\usepackage{shadethm}
\usepackage{subfig}
\usepackage{enumitem}
\usepackage{graphicx}
\usepackage{float} 
\usepackage{tikz}
\usetikzlibrary{shapes,arrows,decorations.pathmorphing,graphs,positioning,backgrounds}
\usepackage{algorithm}
\usepackage{algpseudocode}
\usepackage[square,comma, numbers]{natbib}
\usepackage[automark, headsepline]{scrpage2}
\usepackage{hyperref}
\usepackage{nameref}
\definecolor{TUMblue}{cmyk}{1, .54, .04, .19}
\hypersetup{
	colorlinks   = true, 
	urlcolor     = TUMblue, 
	linkcolor    = TUMblue, 
	citecolor   = TUMblue 
}

\usepackage[
left=1.3in,
right=1.3in,
top=0.7in,
bottom=0.7in,
includeheadfoot
]{geometry}

\newtheoremstyle{new}{12pt}{12pt}{\itshape}{}{\bfseries}{.}{1em}{}
\theoremstyle{new}
\newtheorem{Theorem}{Theorem}
\newtheorem{Corollary}{Corollary}
\newtheorem{Proposition}{Proposition}
\newtheorem{Lemma}{Lemma}

\newtheorem{Example}{Example}
\newtheorem{Remark}{Remark}
\newtheorem{Assumption}{Assumption}

\newcommand{\ind}{\mathds{1}}
\newcommand{\R}{\mathds{R}}

\newcommand{\E}{\mathrm{E}}
\newcommand{\Prob}{\mathrm{Pr}}
\newcommand{\var}{{\mathrm{var}}}

\newcommand{\cov}{{\mathrm{cov}}}

\newcommand{\wh}{\widehat}

\allowdisplaybreaks
\usepackage[left]{lineno}
\newcommand*\patchAmsMathEnvironmentForLineno[1]{%
  \expandafter\let\csname old#1\expandafter\endcsname\csname #1\endcsname
  \expandafter\let\csname oldend#1\expandafter\endcsname\csname end#1\endcsname
  \renewenvironment{#1}%
     {\linenomath\csname old#1\endcsname}%
     {\csname oldend#1\endcsname\endlinenomath}}%
\newcommand*\patchBothAmsMathEnvironmentsForLineno[1]{%
  \patchAmsMathEnvironmentForLineno{#1}%
  \patchAmsMathEnvironmentForLineno{#1*}}%
\AtBeginDocument{%
\patchBothAmsMathEnvironmentsForLineno{equation}%
\patchBothAmsMathEnvironmentsForLineno{align}%
\patchBothAmsMathEnvironmentsForLineno{flalign}%
\patchBothAmsMathEnvironmentsForLineno{alignat}%
\patchBothAmsMathEnvironmentsForLineno{gather}%
\patchBothAmsMathEnvironmentsForLineno{multline}%
}

\begin{document}

	\pagestyle{scrheadings}
	\clearscrheadings
	\lohead{Nonparametric density estimation with simplified vine copulas}
	\rohead{\pagemark}
	\lehead{Thomas Nagler and Claudia Czado}
	\rehead{\pagemark}

	\setcounter{page}{1}
	\pagenumbering{arabic}  
	
\author{Thomas Nagler\footnote{Corresponding author, Department of Mathematics, Technische Universit{\"a}t M{\"u}nchen, Boltzmanstra{\ss}e 3,  85748 Garching (email: \href{mailto:thomas.nagler@tum.de}{thomas.nagler@tum.de})} \, and Claudia Czado\footnote{Department of Mathematics, Technische Universit{\"at} M{\"u}nchen,  Boltzmannstra{\ss}e 3, 85748 Garching Germany (email: \href{mailto:cczado@ma.tum.de}{cczado@ma.tum.de})}}
\title{Evading the curse of dimensionality in nonparametric density estimation with simplified vine copulas}
\date{\hspace{3pt} \normalsize\today}

\maketitle

\begin{abstract} 
\noindent {\bfseries \sffamily Abstract}\\
Practical applications of nonparametric density estimators in more than three dimensions suffer a great deal from the well-known curse of dimensionality: convergence slows down as dimension increases. We show that one can evade the curse of dimensionality by assuming a simplified vine copula model for the dependence between variables. We formulate a general nonparametric estimator for such a model and show under high-level assumptions that the speed of convergence is independent of dimension. We further discuss a particular implementation for which we validate the high-level assumptions and establish its asymptotic normality.
Simulation experiments illustrate a large gain in finite sample performance when the simplifying assumption is at least approximately true. But even when it is severely violated, the vine copula based approach proves advantageous as soon as more than a few variables are involved. Lastly, we give an application of the estimator to a classification problem from astrophysics.  \\[12pt]
\end{abstract}
	
	
	\section{Introduction}

Density estimation is one of the most important problems in nonparametric statistics. Most commonly, nonparametric density estimators are used for exploratory data analysis, but find many further applications in fields such as astrophysics, forensics, or biology \citep{Bock04, Aitken2004, Kie2010}. Many of these applications involve the estimation of multivariate densities. However, most applications so far focus on two- or three-dimensional problems. Furthermore, the persistent interest amongst practitioners is contrasted by a falling tide of methodological contributions in the last two decades. 

A probable reason is the prevalence of the \emph{curse of dimensionality}: due to sparseness of the data, nonparametric density estimators converge more slowly to the true density as dimension increases. Put differently, the number of observations required for sufficiently accurate estimates grows excessively with the dimension. As a result, there is very little benefit from the ever-growing sample sizes in modern data. Section 7.2 in \citep{Scott08} illustrates this phenomenon for a kernel density estimator when the standard Gaussian is the target density: to achieve an accuracy comparable to $n=50$ observations in one dimension, more then $n=10^6$ observations are required in ten dimensions.

In general, this issue cannot be solved: \citet{Stone80} proved that any estimator $\wh f$ that is consistent for the class of $p$ times continuously differentiable $d$-dimensional density functions converges at a rate of at most $n^{-p/(2p + d)}$.  More precisely,
\begin{align*}
	\widehat f(\bm x) = f(\bm x) + O_p(n^{-r}),
\end{align*}
for all densities $f$ of this class and some $r>0$, implies that $r \le p/(2p + d)$. The curse of dimensionality manifests itself in the  $d$ in the denominator. It implies that the optimal convergence rate necessarily decreases in higher dimensions. Thus, to evade the curse of dimensionality, all we can hope for is to find subclasses of densities for which the optimal convergence rate does not depend on $d$. One such subclass is the density functions corresponding to independent variables, which can be estimated as a simple product of univariate density estimates. But the independence assumption is very restrictive. We also want the subclass to be rich and flexible. We will show that simplified vine densities are such a class and provide a useful approximation even when the simplifying assumption is severely violated.

\subsection{Nonparametric density estimation based on simplified vine copulas}

We introduce a  nonparametric density estimator whose convergence speed is independent of the dimension. The estimator is build on the foundation of a simplified vine copula model, where the joint density is decomposed into a product of marginal densities and bivariate copula densities, see, e.g., \citep{Czado10} and Section 3.9 in \citep{Joe14}.

First, we separate the marginal densities and the copula density (which captures the dependence between variables). Let $(X_1, \dots, X_d) \in \mathds{R}^d$ be a random vector with joint distribution $F$ and marginal distributions $F_1, \dots F_d$. Provided densities exist, Sklar's Theorem \citep{Sklar59} allows us to rewrite the joint density $f$ as the product of a copula density $c$ and the marginal densities $f_1, \dots, f_d$: for all $\bm x \in \mathds{R}^d$,
\begin{align*}
	f(\bm x) = c\bigl\{F_1(x_1), \dots, F_d(x_d)\bigr\} \times f_1(x_1) \times \dots \times f_d(x_d), 
\end{align*}
where $c$ is the density of the random vector $\bigl(F_1(X_1), \dots, F_d(X_d)\bigr) \in [0,1]^d$. In order to estimate the joint density $f$, we can therefore obtain estimates of the marginal densities $f_1, \dots, f_d$ and the copula density $c$ separately, and then plug them into the above formula. With respect to the curse of dimensionality, nothing is gained (so far) since  estimation of the copula density is still a $d$-dimensional problem.

A crucial insight is that any $d$-dimensional copula density can be decomposed into a product of $d(d-1)/2$ bivariate (conditional) copula densities \citep{Bedford01}. Equivalently, one can build arbitrary $d$-dimensional copula densities by using  $d(d-1)/2$ building blocks (so-called \emph{pair-copulas}).  Following this idea,  the flexible class of \emph{vine copula} models --- also known as \emph{pair-copula-constructions (PCCs)} --- were introduced in \citep{Aas09} and  have seen rapidly increasing interest in recent years. For instance, a three-dimensional joint density can be decomposed as
\begin{align*}
f(x_1, x_2, x_3) &= c_{1,2}\bigl\{F_1(x_1), F_2(x_2)\bigr\}\times c_{2,3}\bigl\{F_2(x_2), F_3(x_3)\bigr\}  \\
&\phantom{=} \times c_{1,3 ;2}\bigl\{F_{1|2}(x_1|x_2), F_{3|2}(x_3|x_2) \, ; \, x_2\bigr\} \\
&\phantom{=} \times f_1(x_1) \times f_2(x_2) \times f_3(x_3),
\end{align*}
where $c_{1,3;2}\{F_{1|2}(x_1|x_2), F_{3|2}(x_3|x_2) \, ; \, x_2\}$ is the joint density corresponding to the conditional random vector $\bigl(F_{1|2}(X_1|X_2), F_{3|2}(X_3|X_2)\bigr) \big| X_2 = x_2$. Note that the copula of the vector depends on the value $x_2$ of the conditioning variable $X_2$. To reduce the complexity of the model, it is usually assumed that the influence of the conditioning variable on the copula can be ignored. In this case, the conditional density $c_{1,3;2}$ collapses to an unconditional --- and most importantly, two-dimensional --- object, and one speaks of the \emph{simplifying assumption} or a simplified vine copula model/PCC. For general dimension $d$, a similar decomposition into the product of $d$ marginal densities and $d(d-1)/2$ pair-copula densities holds.

Some copula classes where the simplifying assumption is satisfied are given in \citep{Stoeber13}. An important special case is the Gaussian copula. It is the dependence structure underlying a multivariate Gaussian distribution and can be fully characterized by $d(d-1)/2$ partial correlations. Note that under a multivariate Gaussian model, conditional correlations and partial correlations coincide. This property is in direct correspondence to the simplifying assumption which states that all conditional copulas collapse to partial copulas. When the Gaussian copula is represented as a vine copula, it consists of $d(d-1)/2$ Gaussian pair-copulas where the copula parameter of each pair corresponds to the associated partial correlation. In a general simplified vine copula model, we replace each Gaussian pair-copula by an arbitrary bivariate copula. Such models are extremely flexible and encompass a wide range of dependence structures. The class of simplified vine distributions is even more flexible, because it allows to couple a simplified vine copula model with arbitrary marginal distributions.

Under the simplifying assumption, a $d$-dimensional copula density can be decomposed into $d(d-1)/2$ unconditional bivariate densities. Consequently, the estimation of a $d$-dimensional copula density can subdivided into the estimation of $d(d-1)/2$ two-dimensional copula densities. Intuitively, we expect that the convergence rate of such an estimator will be equal to the rate of a two-dimensional estimator and, thus, there is no curse of dimensionality. This is formally established by our main result: \autoref{Theory:rate_thm}.

Nonparametric estimation of simplified vine copula densities has been discussed earlier using kernels \citep{Lopez13} and  smoothing splines \citep{Kauermann14}. However, both contributions lack an analysis of the asymptotic behavior of the estimators.   We treat the more general setting of densities with arbitrary support. \autoref{Theory:rate_thm} shows under high-level conditions that the convergence rate of a nonparametric estimator of a simplified vine density is independent of the dimension --- an extremely powerful property that has been overlooked so far.

\subsection{Organization}

The remainder is structured as follows: \autoref{Vines} gives a review of vine copulas and introduces notation. A general nonparametric estimator of simplified vine densities is described in detail in \autoref{RVineKDE}. In \autoref{Theory} we show under high-level assumptions that such an estimator is consistent and that the convergence rate is independent of the dimension. Hence, there is no curse of dimensionality.  In \autoref{Practical} we discuss how the method can be implemented as a kernel estimator. For this particular implementation, we validate the high-level assumptions of \autoref{Theory:rate_thm} and establish asymptotic normality.  We illustrate its advantages via simulations in the simplified as well as non-simplified setting  (\autoref{Simulations}). The method is applied to a classification problem from astrophysics in \autoref{Application}. We conclude with a discussion of our results and provide links to the existing literature on the simplifying assumption in \autoref{Conclusion}.
	\section{Simplified vine copulas and distributions} \label{Vines}

We will briefly recall the most important facts about vine copulas and the closely related vine distributions. For a more extensive introduction we refer to \citep{Aas09, Czado10} and Chapter 3 of \citep{Joe14}.

Vine copula models follow the idea of \citet{Joe96} that any $d$-dimensional copula can be expressed in terms of $d(d-1)/2$ bivariate (conditional) copulas. Because such a decomposition is not unique, \citep{Bedford02} introduces a graphical method to organize the structure of a $d$-dimensional  vine copula in terms of linked trees $T_m=(V_m, E_m)$, $m=1, \dots, d-1$. 
A sequence $\mathcal{V}:= (T_1, \dots, T_{d-1})$ of trees is called a \emph{regular vine (R-vine) tree sequence} on $d$ elements if the following conditions are satisfied:
	\begin{enumerate}
		\item $T_1$ is a tree with nodes $V_1=\{1, \dots, d\}$ and edges $E_1$.
		\item For $m\ge 2$, $T_m$ is a tree with nodes $V_m=E_{m-1}$ and edges $E_m$.
		\item (\emph{Proximity condition}) Whenever two nodes in $T_{m+1}$ are joined by an edge, the corresponding edges in $T_m$ must share a common node.
	\end{enumerate} 
    
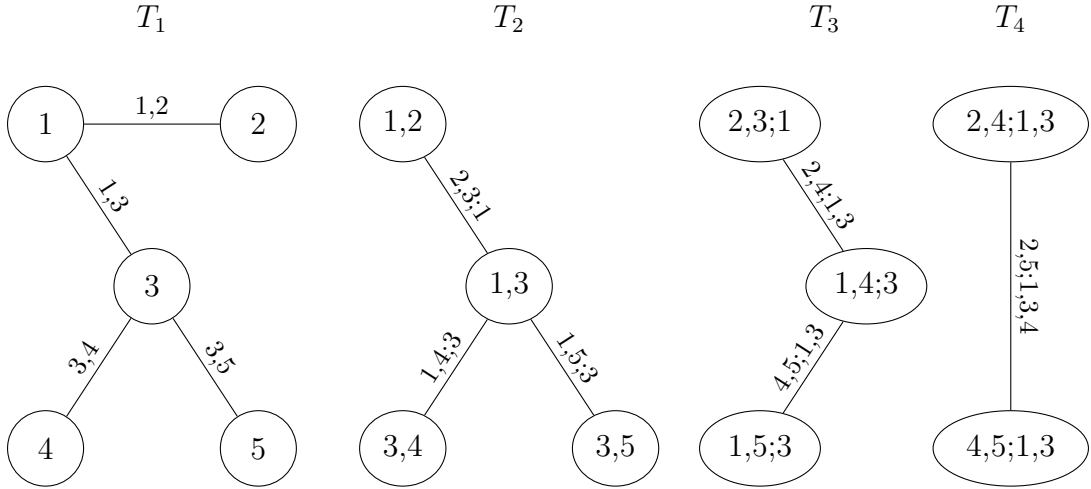
\begin{figure}
\tikzstyle{VineNode} = [ellipse, fill = white, draw = black, text = black, align = center, minimum height = 1cm, minimum width = 1cm]
\tikzstyle{DummyNode}  = [draw = none, fill = none, text = white] 
\tikzstyle{TreeLabels} = [draw = none, fill = none, text = black] 
\newcommand{\labelsize}{\footnotesize} 
\newcommand{\yshiftLabel}{-0.3cm}  
\newcommand{\yshiftNodes}{-0.75cm} 
\newcommand{\xshiftTree}{0.5cm}    
\newcommand{\rotateLabels}{-57}    
	\centering
	\begin{tikzpicture}	[every node/.style = VineNode, node distance =1.4cm] 
	\node (1){1}
	node[DummyNode]  (Dummy12)   [right of = 1]{} 
	node             (2)         [right of = Dummy12] {2}
	node             (3)         [below of = Dummy12, yshift = \yshiftNodes] {3}
	node[DummyNode]  (Dummy45)   [below of = 3, yshift = \yshiftNodes]{}
	node             (4)         [left of = Dummy45] {4}
    node             (5)         [right of = Dummy45] {5}
    node             (12)        [right of = 2, xshift = \xshiftTree] {1,2} %
    node[DummyNode]  (Dummy12x)  [right of = 12]{} 
	node             (13)        [below of = Dummy12x, yshift = \yshiftNodes] {1,3}
	node[DummyNode]  (Dummy45c3) [below of = 13, yshift = \yshiftNodes]{} 
	node             (34)        [left  of = Dummy45c3] {3,4}
    node             (35)        [right of = Dummy45c3] {3,5}
    node             (15c3)      [right of = 35, xshift = \xshiftTree] {1,5;3}
	node[DummyNode]  (Dummy15c3x)[right of = 15c3]{}
	node             (14c3)      [above of = Dummy15c3x, yshift = -\yshiftNodes] {1,4;3}
	node[DummyNode]  (Dummy23c1x)[above of = 14c3, yshift = -\yshiftNodes]{}
	node             (23c1)      [left of  = Dummy23c1x] {2,3;1}	
	node             (24c13)     [right of = Dummy23c1x, xshift = \xshiftTree] {2,4;1,3}
    node             (45c13)     [right of = Dummy15c3x, xshift = \xshiftTree] {4,5;1,3}
	node[TreeLabels] (T1)        [above of = Dummy12] {$T_1$}
	node[TreeLabels] (T2)        [above of = Dummy12x] {$T_2$}
	node[TreeLabels] (T3)        [above of = 23c1] {\hspace{1.7cm}$T_3$} 
	node[TreeLabels] (T4)        [above of = 24c13] {$T_4$}	 
	;	    	
	\draw (1) to node[draw=none, fill = none, font = \labelsize,
	                    above, yshift = \yshiftLabel] {1,2} (2);
	\draw (1) to node[draw=none, fill = none, font = \labelsize, 
	                    rotate = \rotateLabels, above, yshift = \yshiftLabel] {1,3} (3);   
	\draw (3) to node[draw=none, fill = none, font = \labelsize, 
	                    rotate = \rotateLabels, above, yshift = \yshiftLabel] {3,5} (5);  
	\draw (3) to node[draw=none, fill = none, font = \labelsize, 
	                    rotate = -\rotateLabels, above, yshift = \yshiftLabel] {3,4} (4); 
	\draw (12) to node[draw=none, fill = none, font = \labelsize, above, 
	                    rotate = \rotateLabels, above, yshift = \yshiftLabel] {2,3;1} (13);   
	\draw (13) to node[draw=none, fill = none, font = \labelsize, above, 
	                    rotate = \rotateLabels, above, yshift = \yshiftLabel] {1,5;3} (35); 
	\draw (13) to node[draw=none, fill = none, font = \labelsize, above, 
	                    rotate = -\rotateLabels, above, yshift = \yshiftLabel] {1,4;3} (34); 
	\draw (23c1) to node[draw=none, fill = none, font = \labelsize, above, 
	                    rotate = \rotateLabels, above, yshift = \yshiftLabel] {2,4;1,3} (14c3);   
	\draw (14c3) to node[draw=none, fill = none, font = \labelsize, above, 
	                    rotate = -\rotateLabels, above, yshift = \yshiftLabel] {4,5;1,3} (15c3); 
	\draw (24c13) to node[draw=none, fill = none, font = \labelsize, above, 
	                    rotate = -90, above, yshift = \yshiftLabel] {2,5;1,3,4} (45c13);
	\end{tikzpicture}
	\caption{Example of a regular vine tree sequence.}
	\label{Vines:RVine_fig}
\end{figure}

The tree sequence is also called the \emph{structure} of the vine. An example of an R-vine tree sequence for $d=5$ is given in \autoref{Vines:RVine_fig}. For the annotation of the edges in each tree we follow \citep{Czado10}. 

An \emph{R-vine copula} model identifies each edge of the trees with a bivariate copula (a so-called \emph{pair-copula}). Assume that each pair-copula admits a density and let $\mathcal{B}:=\{c_{j_e, k_e ; D_e} | e \in E_m, 1 \le m \le d-1\}$ be the set of copula densities associated with the edges in $\mathcal{V}$. Then, the R-vine copula density can be written as
\begin{align}
c(\bm u) = \prod_{m=1}^{d-1} \prod_{e \in E_m} c_{j_e, k_e; D_e} \bigl\{G_{j_e|D_e}(u_{j_e}|\bm u_{D_e}), \, G_{k_e|D_e}(u_{k_e}|\bm u_{D_e}) ; \, \bm u_{D_e} \bigr\}, \label{Vines:density_nonsimplified_eq}
\end{align}
where $\bm u_{D_e}:=(u_\ell)_{\ell \in D_e}$ is a subvector of $\bm u =(u_1, \dots, u_d) \in [0,1]^d$ and $G_{j_e|D_e}$ is the conditional distribution of $U_{j_e} | \bm U_{D_e} = \bm u_{D_e}$. The set $D_e$ is called \emph{conditioning set} and the indices $j_e, k_e$ form the \emph{conditioned set}. In the first tree the conditioning set $D_e$ is empty, and we define $G_{j_e}(u_{j_e}) := u_{j_e}, G_{k_e}(u_{k_e}) := u_{k_e}$ for notational consistency. For a given edge $e$, the function $c_{j_e, k_e ; D_e}$ is the  copula density associated with the conditional random vector 
\begin{align*}
\bigl(G_{j_e|D_e}(U_{j_e}|\bm U_{D_e}),G_{k_e|D_e}(U_{k_e}|\bm U_{D_e})\bigr)\bigl| \bm U_{D_e} = \bm u_{D_e}. 
\end{align*}

Note that in \eqref{Vines:density_nonsimplified_eq}, the pair-copula density $c_{j_e, k_e; D_e}$ takes $\bm u_{D_e}$ as an argument and the functional form w.r.t.\ the arguments $u_{j_e}$,$u_{k_e}$ may be different for each value of $\bm u_{D_e}$. This conditional structure makes the model very complex and complicates estimation. To simplify matters, we assume that this dependence can be ignored and the copula is equal across all possible values of $\bm u_{D_e}$: we assume that the \emph{simplifying assumption} holds. In this case, \eqref{Vines:density_nonsimplified_eq} collapses to
\begin{align}
c(\bm u) = \prod_{m=1}^{d-1} \prod_{e \in E_m} c_{j_e, k_e; D_e} \bigl\{G_{j_e|D_e}(u_{j_e}|\bm u_{D_e}), \, G_{k_e|D_e}(u_{k_e}|\bm u_{D_e}) \bigr\}. \label{Vines:density_eq}
\end{align} 
A distribution whose copula density can be represented this way is called a \emph{simplified vine distribution}.
\begin{Example}   \label{Vines:vine_ex}
	The density of a simplified R-vine copula corresponding to the tree sequence in \autoref{Vines:RVine_fig} is
	\begin{align*}
	c(u_1, \dots, u_5) &= c_{1,2}(u_1, u_2) \times c_{1,3}(u_1,u_3) \times c_{3,4}(u_3,u_4) \times c_{3,5}(u_3,u_5)\\
	& \phantom{=} \times c_{2,3;1}(u_{2|1}, u_{3|1}) \times c_{1,4;3}(u_{1|3}, u_{4|3}) \times c_{1,5;3}(u_{1|3}, u_{5|3}) \\
	& \phantom{=} \times c_{2,4;1,3}(u_{2|1,3}, u_{4|1,3}) \times c_{4,5;1,3}(u_{4|1,3}, u_{5|1,3}) \\
	& \phantom{=} \times c_{2,5;1,3,4}(u_{2|1,3,4}, u_{5|1,3,4}),
	\end{align*}
	where we used the abbreviation $u_{j_e|D_e} := G_{j_e|D_e}(u_{j_e}|\bm u_{D_e})$.
\end{Example} \noindent
R-vine copula densities involve conditional distributions $G_{j_e|D_e}$. We can express them in terms of conditional distributions corresponding to bivariate copulas in $\mathcal{B}$ as follows: Let $\ell_e \in D_e$ be another index such that $c_{j_e, \ell_e; D_e \setminus \ell_e} \in \mathcal{B}$ and define $D'_e:= D _e \setminus \ell_e$. Then, we can write
\begin{align}
\begin{aligned}
& \; G_{j_e|D_e}(u_{j_e}|\bm u_{D_e}) = h_{j_e|\ell_e;D'_e }\bigl\{G_{j_e|D'_e}(u_{j_e}|\bm u_{D'_e})\, \big| \, G_{\ell_e|D'_e}(u_{\ell_e}|\bm u_{D'_e})\bigr\}, 
\end{aligned}
 \label{Vines:h_recursive_eq}
\end{align}
where the \emph{h-function} is defined as
\begin{align}
h_{j_e|\ell_e;D'_e} (u | v) := \int_{0}^{u} c_{j_e, \ell_e;D'_e}(s , v) ds, \qquad \mbox{for } (u,v) \in [0,1]^2. \label{Vines:hfuncdef_eq}
\end{align}
By definition, h-functions are conditional distribution functions for pairs of marginally uniformly distributed random variables with joint density $c_{j_e, \ell_e;D'_e}$. The arguments $G_{j_e|D'_e}(u_{j_e}|\bm u_{D'_e})$ and $G_{\ell_e|D'_e}(u_{\ell_e}|\bm u_{D'_e})$ of the h-function in \eqref{Vines:h_recursive_eq} can be rewritten in the same manner. In each step of this recursion the conditioning set $D_e$ is reduced by one element. Note also that, by construction, the copula density on the right hand side of \eqref{Vines:hfuncdef_eq} always belongs to the set $\mathcal{B}$. Eventually, this allows us to write any of the conditional distributions $G_{j_e|D_e}$ as a recursion over h-functions that are directly linked to the pair-copula densities. Later, we will use this fact to derive estimates of such conditional distributions from estimates of the pair-copula densities in lower trees.
\begin{Example}
	Consider an R-vine copula corresponding to the R-vine tree sequence given in \autoref{Vines:RVine_fig}. We have 
	\begin{align*}
	G_{3|1,2}(u_3| u_1, u_2) &=  h_{3|2;1}\bigl\{h_{3|1}(u_3|u_1) \big| h_{2|1}(u_2|u_1) \bigr\},
	\end{align*}	
	where $h_{3|2;1}(u_{3|1}|u_{2|1}) = \int_{0}^{u_{3|1}} c_{2,3;1}(u_{2|1}, s) ds$, $h_{3|1}(u_3|u_1) = \int_{0}^{u_3} c_{1,3}(u_1, s) ds$, and $h_{2|1}(u_2|u_1) = \int_0^{u_2} c_{1,2}(u_1, s) ds$.
\end{Example} \noindent
Altogether, we can express any vine copula density in terms of bivariate copula densities and corresponding h-functions.
	\section{A nonparametric density estimator based on simplifed vine copulas} \label{RVineKDE}

We propose a multivariate nonparametric density estimation technique where a) we separate the estimation of marginal and copula densities, and b) the copula density is estimated as the product of sequentially estimated pair-copula densities. We suggest a general step-wise estimation algorithm without specifying exactly how the components are estimated. This more practical issue is deferred to \autoref{Practical}.

Let $\bm X = (X_1, \dots, X_d) \in \Omega_{\bm X}$ be a random vector with continuous joint distribution $F$ and  marginal distributions $F_1, \dots, F_d$. The support of $X_\ell$ will be denoted as $\Omega_{X_\ell}$, $\ell = 1, \dots, d$. Let further $\bm X^{(i)} = (X_1^{(i)}, \dots, X_d^{(i)})$, $i=1, \dots, n$, be \emph{iid} copies  of $\bm X$ (acting as observations).  Assume that $F$ is a simplified vine distribution with structure  $\mathcal{V} =(T_1, \dots, T_{d-1})$. Provided densities exist, we can use Sklar's theorem and \eqref{Vines:density_eq} to write the joint density $f$ for all $\bm x = (x_1, \dots, x_d) \in \Omega_{\bm X}$ as
\begin{align} \label{RVineKDE:Rvinedensity_eq}
f(\bm x) &=  c\bigl\{F_1(x_1), \dots, F_d(x_d)\bigr\} \times \prod_{l=1}^d f_\ell(x_\ell) \notag \\
\begin{split}
&= \prod_{m=1}^{d-1} \prod_{e \in E_m} c_{j_e, k_e; D_e} \bigl\{F_{j_e|D_e}(x_{j_e}|\bm x_{D_e}), \, F_{k_e|D_e}(x_{k_e}|\bm x_{D_e}) \bigr\}  \times \prod_{l=1}^d f_\ell(x_\ell).
\end{split}
\end{align}
The conditional distribution functions $F_{k_e|D_e}(x_{k_e}|\bm x_{D_e})$ can equivalently be expressed as $G_{k_e|D_e}(u_{k_e}|\bm u_{D_e})$, where $\bm u = (u_1, \dots, u_d) := (F_1(x_1), \dots, F_d(x_d))$. This allows us to decompose  $F_{k_e|D_e}$ recursively into h-functions (see \autoref{Vines}).

The idea is now to estimate all functions in the above expression separately. We use a step-wise estimation procedure that is widely used in vine copula models, see, e.g., \citep{Aas09, Haff13}. It is summarized in \autoref{RVineKDE:seqest_alg}. Let us describe the reasoning behind the first few steps in a little more detail.

\begin{enumerate}[label=\arabic*.]
	\item Based on the observations $(X_1^{(i)}, \dots, X_d^{(i)})$, $i =1, \dots, n,$ we obtain  estimates $\wh f_1, \dots, \wh f_d, \wh F_1, \dots, \wh F_d$ of the marginal densities $f_1, \dots, f_d$ and distribution functions $F_1, \dots, F_d$.

\item The copula density $c$ is the density of the random vector $\bm U :=\bigl(F_1(X_1), \dots, F_d(X_d)\bigr)$. We do not have access to observations from this vector. However, we can define pseudo-observations $\bm U^{(i)} :=\bigl(\wh U_1^{(i)}, \dots, \wh U_d^{(i)}\bigr)$ by replacing $F_1, \dots, F_d$ with the estimators from the last step:
	\begin{align}
		\bigl(\wh U_1^{(i)}, \dots, \wh U_d^{(i)}\bigr) := \bigl(\wh F_1(X_1^{(i)}), \dots, \wh F_d(X_d^{(i)})\bigr), \qquad  i = 1, \dots, n. \label{RVineKDE:pobs_eq}
	\end{align}
	Based on two-dimensional subvectors of the pseudo-observations  \eqref{RVineKDE:pobs_eq}, we estimate all pair-copula densities and h-functions that correspond to edges of the first tree (the conditioning sets $D_e$ are empty).  We use  \eqref{Vines:hfuncdef_eq} to derive estimates of the $h$-functions, that is
	\begin{align*}
	\wh h_{j_e|k_e} (u | v) := \int_{0}^{u} \wh c_{j_e, k_e}(s , v) ds, \qquad \mbox{for } (u,v) \in (0,1)^2. 
	\end{align*}
	Optionally, the h-functions can be estimated separately. However, this will typically lead to a density estimate that does not integrate to one.
    
	\item Any pair-copula density $c_{j_e, k_e; D_e}$ corresponding to an edge in the second tree is the density of a random vector $\bigl(F_{j_e|D_e}(X_{j_e}| X_{D_e}), F_{k_e|D_e}(X_{k_e}| X_{D_e})\bigr)$, $e \in E_2$. They are not observable, but  we can use pseudo-observations such as
	\begin{align*}
	\wh U_{j_e|D_e }^{(i)} := \wh F_{j_e|D_e}\bigl( X_{j_e}^{(i)}| X_{D_e}^{(i)}\bigr) =  \wh G_{j_e|D_e}\bigl(\wh U_{j_e}^{(i)}| \wh U_{D_e}^{(i)}\bigr)=  \wh h_{j_e|D_e}\bigl(\wh U_{j_e}^{(i)}| \wh U_{D_e}^{(i)}\bigr), 
	\end{align*}
	$i=1,\dots, n,$ instead. This allows us to obtain estimates  $\wh c_{j_e, k_e; D_e}$, $\wh h_{j_e| k_e; D_e}$, and $\wh h_{k_e| j_e; D_e}$.
	
	\item For estimation in the third tree, we need observations from random vectors such as
	\begin{align}
	U_{j_e|D_e }^{(i)} := F_{j_e|D_e}\bigl(X_{j_e}^{(i)}| {\bm{X}}_{D_e}^{(i)}\bigr), \label{RVineKDE:rvs3_eq}
	\end{align}
	$i=1,\dots, n$, $e \in E_3$. Recall from \autoref{Vines} that, by construction, we can find some edge $e^\prime \in E_2$ such that $j_{e^\prime} = j_e$ and $D_{e^\prime} \cup k_{e^\prime} = D_e$. Consequently, we can apply \eqref{Vines:h_recursive_eq} and approximate \eqref{RVineKDE:rvs3_eq} by the pseudo-observations
	\begin{align*} 
		\begin{aligned}
			\wh U_{j_e|D_e}^{(i)}  = \wh U_{j(e^\prime)|D(e^\prime) \cup k(e^\prime)}^{(i)} &:= \wh F_{j_{e^\prime}|D_{e^\prime} \cup k_{e^\prime}}\bigl(X_{j_{e^\prime}}^{(i)} \big| \bm X_{D_{e^\prime} \cup k_{e^\prime}}^{(i)} \bigr) \\
            &\phantom{:}=\wh G_{j_{e^\prime}|D_{e^\prime} \cup k_{e^\prime}}\bigl(\wh U_{j_{e^\prime}}^{(i)} \big| \wh{\bm U}_{D_{e^\prime} \cup k_{e^\prime}}^{(i)} \bigr) \\
			&\phantom{:}= \wh h_{j_{e^\prime}|k_{e^\prime} ; D_{e^\prime}} \bigl(\wh U_{j_{e^\prime}|D_{e^\prime}}^{(i)}\big| \wh{ U}_{k_{e^\prime} | D_{e^\prime}}^{(i)}\bigr),
		\end{aligned}
	\end{align*}
	where the last equality is again derived from \eqref{Vines:h_recursive_eq}.
	
	\item For higher trees, proceed as in step 4.\
\end{enumerate}
	 At the end of the procedure we have estimates for all marginal distributions/densities, bivariate copula densities, and all h-functions that are required to evaluate the R-vine density \eqref{RVineKDE:Rvinedensity_eq}. For all $\bm x \in \Omega_{\bm X}$ we now define an estimate of the simplified vine density $f$ as
	\begin{align} 
		\wh f_{\mathrm{vine}}(\bm x) &:= \prod_{m=1}^{d-1} \prod_{e \in E_m} \wh c_{j_e, k_e; D_e} \bigl\{\wh F_{j_e|D_e}( x_{j_e}|{\bm x}_{D_e}), \, \wh F_{k_e|D_e}(x_{k_e}|{\bm x}_{D_e}) \bigr\} \times \prod_{\ell=1}^d \wh f_\ell(x_\ell).  \label{RVineKDE:estimator_eq}
	\end{align}

\begin{algorithm}[t]
	\caption{Sequential estimation of simplified vine densities}
	\label{RVineKDE:seqest_alg}
	{\bfseries Input:} Observations $(X_1^{(i)}, \dots, X_d^{(i)})$, $i=1, \dots, n$, structure $\mathcal{V}=(T_1, \dots, T_{d-1})$. \\
	{\bfseries Output:} Estimates of all marginal densities and distributions, pair-copula densities, and h-functions required to evaluate the simplified vine density \eqref{RVineKDE:Rvinedensity_eq}.\\	---------------------------------------------------------------------------------------------------------\\
	{\bfseries for} $\ell=1, \dots, d$:\\
	\hspace*{2em} Obtain estimates $\wh f_\ell, \wh F_\ell$ of the marginal density $f_\ell$ and distribution $F_\ell$.\\
	\hspace*{2em} Set $\wh U_\ell^{(i)} := \wh F_\ell(X_\ell^{(i)}), i=1, \dots, n$.\\
	{\bfseries end for}\\
	{\bfseries for} $m=1, \dots, d-1$:\\
	\hspace*{2em} {\bfseries for all} $e \in E_m$:\\[-12pt]
	\hspace*{1em}
	\begin{minipage}[t]{0.9\textwidth}
		\begin{enumerate}
			\item {\itshape Estimation step:} Based on $\bigl(\wh U_{j_e|D_e}^{(i)}, \wh U_{k_e|D_e}^{(i)}\bigr)_{i=1,\dots,n}$, obtain an estimate of the copula density $c_{j_e, k_e;D_e}$  which we denote as $\wh c_{j_e, k_e;D_e}$, and corresponding h-function estimates $\wh h_{j_e| k_e;D_e}$, $\wh h_{k_e| j_e;D_e}$.
			\item {\itshape Transformation step:} Set
			\begin{align*}
			\wh U_{j_e|D_e \cup k_e}^{(i)} &:= \wh h_{j_e|k_e ; D_e}\bigl(\wh U_{j_e|D_e}^{(i)}\big| \wh{ U}_{k_e | D_e}^{(i)}\bigr), \\
			\wh U_{k_e|D_e \cup j_e}^{(i)} &:= \wh h_{k_e|j_e ; D_e}\bigl(\wh U_{k_e|D_e}^{(i)}\big| \wh{U}_{ j_e|D_e}^{(i)}\bigr), \quad i = 1, \dots, n.
			\end{align*}
		\end{enumerate}
	\end{minipage}
	\hspace*{2em} {\bfseries end for}\\
	{\bfseries end for}
\end{algorithm}
	\section{Asymptotic theory} \label{Theory}

We now establish weak consistency of the simplified vine density estimator proposed in \autoref{RVineKDE}. We furthermore show that its probabilistic convergence rate does not increase with dimension and, hence, there is no curse of dimensionality. 

\subsection{Consistency and rate of convergence} \label{Theory:rates}

The sequential nature of the proposed estimator complicates its analysis. Estimation errors will propagate from one tree to the next and affect the estimation in higher trees. We impose high-level assumptions on the uni- and bivariate estimators that allow us to establish our main result. 

The first assumption considers the consistency of univariate density and distribution function estimators. Although estimators may converge at different rates, we will formulate all assumptions w.r.t.\ to the same rate $n^{-r}$, $r > 0$. This rate then has to be the slowest among all estimators involved --- typically the rate of the pair-copula density estimator. 
\begin{Assumption} \label{Theory:f_ass}
For all $\ell = 1, \dots, d,$ and all $x_\ell \in \Omega_{X_\ell}$, it holds
\begin{align*}
		(a) \quad \wh f_\ell(x_\ell) - f_\ell(x_\ell) = O_p(n^{-r}), \qquad (b)  \quad \sup_{x_\ell\in \Omega_{X_\ell}} \bigl\vert\wh F_\ell(x_\ell) - F_\ell(x_\ell)\bigr\vert = o_{a.s.}(n^{-r}).
\end{align*}
\end{Assumption}\noindent
Next, assume  we are in an ideal situation where, for each edge $e \in E_m, m = 1, \dots, d-1$, we have access to the true (but unobservable) pair-copula samples
\begin{align} \label{Theory:unobs_eq}
\begin{aligned}
U_{j_e|D_e}^{(i)} := F_{j_e|D_e} \bigl(X_{j_e}^{(i)}| \bm X_{D_e}^{(i)}\bigr),  \qquad
U_{k_e|D_e}^{(i)} := F_{k_e|D_e}\bigl(X_{k_e}^{(i)}| \bm X_{D_e}^{(i)}\bigr),
\end{aligned} 
\end{align}
$i=1, \dots,n,$. Recall that estimators are functions of the data, although this dependence is usually not made explicit in notation. Denote
\begin{align}  \label{oracle_def}
\overline c_{j_e,k_e;D_e} (u,v) := \overline c_{j_e,k_e;D_e}\bigl(u, v, U_{j_e|D_e}^{(1)}, \dots, U_{k_e|D_e}^{(n)} \bigr)
\end{align}
 as the oracle pair-copula density estimator that is based on the random samples \eqref{Theory:unobs_eq}.  The  h-function estimators corresponding to \eqref{oracle_def} are denoted $\overline h_{j_e|k_e;D_e}$ and $\overline h_{k_e|j_e;D_e}$. The second assumption requires the pair-copula density and h-function estimators to be consistent in this ideal world. For the h-functions we need strong uniform consistency on compact interior subsets of $[0,1]^2$. We further assume that the errors from h-function estimation vanish faster than $n^{-r}$. 
\begin{Assumption} \label{Theory:oracle_ass}
	For all $e \in E_m, m = 1, \dots, d-1$, it holds:
    \begin{enumerate}
    \item[$(a)$] for all $(u,v) \in (0, 1)^2$,
    \begin{align*}
	 \overline c_{j_e,k_e;D_e}(u,v) - c_{j_e,k_e;D_e}(u,v)  = O_p(n^{-r}),
	\end{align*}
    
    \item[$(b)$] for every $\delta \in (0, 0.5]$,
    \begin{align*}
&  \sup_{(u, v) \in [\delta, 1-\delta]^2} \bigl\vert \overline h_{j_e|k_e;D_e}(u|v) - h_{j_e|k_e;D_e}(u|v) \bigr\vert = o_{a.s.}(n^{-r}), \\ 
&  \sup_{(u, v) \in [\delta, 1-\delta]^2}  \bigl\vert  \overline h_{k_e|j_e;D_e}(u|v) - h_{k_e|j_e;D_e}(u|v) \bigr\vert = o_{a.s.}(n^{-r}).
\end{align*}
\end{enumerate}
\end{Assumption}\noindent

In practice, one has to replace \eqref{Theory:unobs_eq} by pseudo-observations which have to be estimated. Thus, we only have access to perturbed versions of the random variables \eqref{Theory:unobs_eq}. Similar to a Lipschitz condition, the last assumption ensures that the pair-copula and h-function estimators are not overly sensitive to such perturbations. Denote 
\begin{align}  \label{feasible_def}
\wh c_{j_e,k_e;D_e} (u,v) := \overline c_{j_e,k_e;D_e}\bigl(u, v, \wh U_{j_e|D_e}^{(1)}, \dots, \wh U_{k_e|D_e}^{(n)} \bigr)
\end{align}
as the estimator based on pseudo-observations $\wh U_{j_e|D_e}^{(i)}, \wh U_{k_e|D_e}^{(i)}$ (as defined in \autoref{RVineKDE:seqest_alg}). The h-function estimators corresponding to \eqref{feasible_def} are denoted $\wh h_{j_e|k_e;D_e}$ and $\wh h_{k_e|j_e;D_e}$. 
\begin{Assumption} \label{Theory:bound_ass}
	For all $e \in E_m, m = 1, \dots, d-1$, it holds:
    \begin{enumerate}
	\item[$(a)$] for all $(u,v) \in (0,1)^2$,
    \begin{align*}
		 \wh c_{j_e,k_e;D_e}(u, v) - \overline c_{j_e,k_e;D_e}(u, v) = O_p(a_{e, n}),
    \end{align*}
 
    \item[$(b)$] for every $\delta \in (0, 0.5]$, 
	\begin{align*}
       & \sup_{(u, v) \in [\delta, 1-\delta]^2} \bigl\vert \wh h_{j_e|k_e;D_e}(u|v) - \overline h_{j_e|k_e;D_e}(u|v) \bigr\vert = O_{a.s.}(a_{e,n}), \\ 
     & \sup_{(u, v) \in [\delta, 1-\delta]^2}  \bigl\vert  \wh h_{k_e|j_e;D_e}(u|v) - \overline h_{k_e|j_e;D_e}(u|v)\bigr\vert = O_{a.s.}(a_{e,n}),
	\end{align*}
\end{enumerate}
where
    \begin{align*}
      a_{e, n} :=  \sup_{i = 1, \dots, n} \vert\wh U_{j_e|D_e}^{(i)} - U_{j_e|D_e}^{(i)}\bigr\vert +  \bigl\vert\wh U_{k_e|D_e}^{(i)} - U_{k_e|D_e}^{(i)}\bigr\vert.
	\end{align*}
\end{Assumption} \noindent
Finally, we require the true pair-copula densities to be smooth. Note that smoothness of pair-copula densities already guarantees smoothness of related h-functions by \eqref{Vines:hfuncdef_eq}. 

\begin{Assumption} \label{Theory:smoothness_ass}
For all $e \in E_m$, $m = 1, \dots, d-1$, the pair-copula densities $c_{j_e,k_e;D_e}$ are continuously differentiable on $(0, 1)^2$.
\end{Assumption}\noindent 
Now we can state our theorem. The proof is deferred to \autoref{Appendix}.

\begin{Theorem} \label{Theory:rate_thm}
Let $f$ be a $d$-dimensional density corresponding to a simplified vine distribution with structure $\mathcal{V}=(T_1, \dots, T_{d-1})$ and let $(X_1^{(i)}, \dots, X_d^{(i)}), i=1, \dots, n$, be $iid$ observations from this density. Denote further $\wh f_{\mathrm{vine}}$ as the estimator resulting from \autoref{RVineKDE:seqest_alg} with $(X_1^{(i)}, \dots, X_d^{(i)})_{i = 1, \dots, n}$ and $\mathcal{V}$ as the input.
Under \hyperref[Theory:f_ass]{Assumptions \ref*{Theory:f_ass}}\hyperref[Theory:smoothness_ass]{--}\ref{Theory:smoothness_ass}, it holds for all $\bm x \in \Omega_{\bm X}$,
	\begin{align*}
	\wh f_{\mathrm{vine}}(\bm x) - f(\bm x) = O_p(n^{-r}).
	\end{align*}
\end{Theorem}\noindent

Usually, convergence of nonparametric density estimators slows down as dimension increases. This phenomenon is widely known as the curse of dimensionality and restricts the practical application of the estimators to very low-dimensional problems. By \hyperref[Theory:rate_thm]{Theorem \ref*{Theory:rate_thm}}, the proposed vine copula based kernel density estimator inherits the convergence rate of the bivariate copula density  estimator. It does not depend on the dimension $d$ and, therefore, suffers no curse of dimensionality. This is a direct consequence of the simplifying assumption allowing us to subdivide the $d$-dimensional estimation problem into several one- and two-dimensional tasks.

Assuming that the pair-copula densities are $p$ times continuously differentiable, we can achieve convergence with $r = p/(2p + 2)$. Recalling from \citep{Stone80} that a  general nonparametric density estimator has optimal rate $p/(2p + d)$, we see that the vine copula based estimator converges at a rate that is equivalent to the rate of a two-dimensional classical estimator. As this property is independent of dimension, we can expect large benefits of the vine copula approach especially in higher dimensions. We emphasize that a necessary condition for \autoref{Theory:rate_thm} to hold with $r = p/(2p + 2)$ is that the density $f$ belongs to the class of simplified vine densities. If this is not the case, the estimator described in \autoref{RVineKDE} is not consistent, but converges towards a simplified vine density that is merely an approximation of the true density. More specifically, its limit is the \emph{partial vine copula approximation}, first defined in \citep{Spanhel15}. In \autoref{Simulations} we will illustrate that even in this situation an estimator based on simplified vine copulas  can outperform the classical approach on finite samples.

\begin{Remark}
	\autoref{Theory:rate_thm} allows for densities $f$ with arbitrary support. Their support, $\Omega_{\bm X}$, only relates to the marginal distributions; copulas are always supported on $[0,1]^d$. If some of the $X_\ell$ have bounded support, we just have to use estimators for $\wh f_\ell$ that takes  this into account. This underlines how flexible the vine copula based approach is.
\end{Remark}

\begin{Remark}
 It is straightforward to extend \autoref{Theory:rate_thm} to non-simplified vine densities by extending the pair-copula densities to functions of more than two variables. Besides that, the proof given in \autoref{Appendix} does not make use of the simplifying assumption at all. However, the simplifying assumption is necessary for $r = p/(2p + 2)$ to be feasible. More generally, if we assume that the pair-copulas depend on at most $d^\prime$ conditioning variables, the optimal rate is $p/(2p + 2 + d^\prime)$.
\end{Remark}

\begin{Remark}
 \autoref{Theory:rate_thm} can be extended to 
 \begin{align*}
	\sup_{\bm x \in \Omega_{\bm X}} \bigl\vert \wh f_{\mathrm{vine}}(\bm x) - f(\bm x) \bigr\vert = O_p\bigl\{(\ln n/n)^r\bigr\},
\end{align*}
provided that the rate $n^{-r}$ in our assumptions is replaced by $(\ln n/n)^r$ and holds uniformly on $\Omega_{X_\ell}$ and $[0, 1]^2$ respectively. But this requires that the pair-copula densities are bounded which is unusual. For example, it does not hold when $f$ is a multivariate Gaussian density with non-diagonal covariance matrix. If the assumptions are met, $\wh f_{\mathrm{vine}}$ is able to achieve the optimal uniform rate of a two-dimensional nonparametric density estimator which is attained at $r = p/(2p + 2)$ \citep[see,][]{Stone83}. 
\end{Remark}

\hyperref[Theory:f_ass]{Assumptions \ref*{Theory:f_ass}}\hyperref[Theory:bound_ass]{--}\ref{Theory:bound_ass} are very general and hold for a large class of estimators under mild regularity conditions. In \autoref{Practical} we validate them for a particular implementation which will be used in the simulations (\autoref{Simulations}).


\subsection{A note on the asymptotic distribution} \label{sec:asdistr}

We also want to give a brief and general account of the asymptotic distribution of the estimator. Let $d^* = d + d(d-1)/2$ and $\wh{\bm f}^*(\bm x) \in \R^{d^*}$ be the stacked vector of all components of the product $\wh f_{\mathrm{vine}}(\bm x)$ in Eq.\ \eqref{RVineKDE:estimator_eq}, i.e., 
\begin{align*}
\wh{\bm f}^*(\bm x) := \bigl(\wh f_1(x_1), \wh f_2(x_2), \dots, \wh c_{j_e,k_e|D_e}\bigl\{\wh F_{j_e|D_e}(x_{j_e} | \bm x_{D_e}), \wh F_{k_e|D_e}(x_{k_e} | \bm x_{D_e})\bigr\}, \dots\bigr),
\end{align*}
and similarly ${\bm f^*}(\bm x)$. Then $\prod_{k=1}^{d^*} \wh  f^*_k = \wh f_{\mathrm{vine}}(\bm x)$ and $\prod_{k=1}^{d^*} f^*_k = f(\bm x)$. The following result is a simple application of the multivariate delta method.
\begin{Proposition} \label{Theory:norm_prop}
	If for some $\bm \mu_{\bm x} \in \R^{d^*}$, $ \Sigma_{\bm x} \in \R^{d^*\times d^*},$
	\begin{align} \label{Theory:norm_eq}
	n^{r}\bigl\{\wh{\bm f}^*(\bm x) - {\bm f^*}(\bm x) \bigr\} \stackrel{d}{\to}\mathcal{N}_{d^*}\bigl(\bm \mu_{\bm x}, \Sigma_{\bm x}\bigr),
	\end{align}
	then for all $\bm x \in \R^d$,
	\begin{align*}
	n^{r}\bigl\{\wh f_{\mathrm{vine}}(\bm x) -  f(\bm x)\bigr\} \stackrel{d}{\to} \mathcal{N}_d\bigl(\bm \theta^\top \bm \mu_{\bm x}, \bm \theta^\top \Sigma_{\bm x} \bm \theta\bigr),
	\end{align*}
	where $\theta_k = \prod_{j \neq k}  f^*_j(\bm x)$, $k = 1, \dots, d^*$.
\end{Proposition}

The standard way to establish the joint normality assumption \eqref{Theory:norm_eq} is to check the conditions of the multivariate Lindeberg-Feller central limit theorem (see Proposition 2.27  of \citep{vanderVaart98}). We will do this for a particular implementation in \autoref{Practical} (see \autoref{prop:an}).
    \section{On an implementation as kernel estimator} \label{Practical}
So far we did not specify how the marginal densities, pair-copula densities, and h-functions should be estimated. In general, we can tap into the full potential of existing methods. In this section, we discuss a particular implementation as a kernel estimator. We give low-level conditions under which the assumptions of \autoref{Theory:rate_thm} can be verified. We present corresponding consistency results and establish asymptotic normality of $\wh f_{\mathrm{vine}}$. Similar results could be obtained for other implementations. Another issue is that we assumed the structure of the vine to be known. Some heuristics to select an appropriate vine structure are discussed at the end of this section.

\subsection{Estimation of marginal densities and distribution functions} \label{UnivKDE}

Univariate kernel density and distribution function estimators have been extensively studied in the literature. To this day, they are most popular in their original form \citep{Rosenblatt56, Parzen62}: for all $x \in \mathds{R}$,
\begin{align} \label{def:fhat}
	\wh f_\ell(x) = \frac 1 {n b_n}\sum_{i=1}^{n} K\biggl(\frac{X_\ell^{(i)} - x}{b_n}\biggr),\quad  \wh F_\ell(x) = \frac 1 {n}\sum_{i=1}^{n} J\biggl(\frac{X_\ell^{(i)} - x}{b_n}\biggr), 
\end{align}
where $b_n>  0$ is the bandwidth parameter, $K$ is a kernel function and $J(x) = \int_{-\infty}^x K(s)ds$ the integrated kernel. We impose the following assumptions on the kernel function, bandwidth sequence, and marginal distributions.
  \begin{enumerate}[label=K\arabic*:,ref=K\arabic*] 
		\item The kernel function $K$ is a symmetric probability density function supported on $[-1,1]$ and has continuous first-order derivative. \label{K1}
        \item The bandwidth sequence satisfies  $b_n \to 0$ and $nb_n^4/\ln n \to \infty$.  \label{K2}
\end{enumerate}
  \begin{enumerate}[label=M\arabic*:,ref=M\arabic*] 
        \item For all $\ell = 1\, \dots, d$, $f_\ell$ is strictly positive on $\R$ and has uniformly continuous second-order derivative. \label{M1}
	\end{enumerate}
The following result gives the rate of strong uniform consistency for $\wh f_\ell$.
\begin{Proposition}\label{prop:fhat}
	Under conditions \ref{K1}, \ref{K2}, and \ref{M1}, the estimator \eqref{def:fhat} satisfies
    \begin{align*}
		\sup_{x \in \R}\bigl\vert \wh f_\ell(x) - f_\ell(x) \bigr\vert = O_{a.s.}\bigl(b_n^2 + \sqrt{\ln n / (nb_n)}\bigr).
	\end{align*}
     for all $\ell = 1\, \dots, d$.
\end{Proposition}\noindent
\begin{proof}
	A standard result for kernel density estimation \citep[see, e.g.,][Section 6.2.1]{Scott08} is
   \begin{align*}
		\E\bigl\{\wh f_\ell(x)\bigr\} - f_\ell(x) = \frac{1}{2} b_n^2 \sigma^2_K \frac{\partial^2}{\partial x^2} f_\ell(x) + o(b_n^2), 
\end{align*}
where $\sigma_K^2 = \int_{[-1, 1]}x^2K(x)dx < \infty$ by \ref{K1} and $\partial^2/\partial x^2 f_\ell(x)$ is bounded by \ref{M1}. The claim then follows from Theorem 2.3 of \citep{Gine02} which states
\begin{align}
	\sup_{x \in \R} \bigl\vert \wh f_\ell(x) - \E\bigl\{\wh f_\ell(x)\bigr\} \bigr\vert = O_{a.s.}\bigl(\sqrt{\ln n / (nb_n)}\bigr). \tag*{\qedhere}
\end{align}
\end{proof}\noindent
\autoref{prop:fhat} implies pointwise weak consistency of $\wh f_\ell$ as well as strong uniform consistency of $\wh F_\ell$ with the same rate. In both cases the rate could be improved, but the result will be sufficient for our purposes. The mean-square optimal bandwidth for $\wh f_\ell$ is $b_n = O(n^{-1/5})$ for which  \autoref{prop:fhat} holds with rate $O_{a.s.}(n^{-2/5}\sqrt{\ln n})$. 

Extensions of the above estimator comprise variable bandwidth methods \citep{Sain96}, transformation techniques for heavy-tailed distributions \citep{Bolance08}, and boundary kernel estimators that avoid bias and consistency issues on bounded support \citep{Bouezmarni10}. 


\subsection{Estimation of pair-copula densities} \label{BiCopKDE}

Nonparametric estimation of copula densities requires caution because they are  supported on  the unit hypercube. An estimator that  takes no account of this property will suffer from bias issues at the boundaries of the support. A few kernel estimators particularly suited for bivariate copula densities were proposed in the literature \citep{Gijbels90, Charpentier06, Geenens14a}. Other nonparametric estimators can be constructed based on Bernstein polynomials \citep{Sancetta04}, B-splines \citep{Kauermann13}, or wavelets \citep{Genest09}. 

In this paper, we will use the transformation estimator of \citep{Charpentier06}. The idea is to transform the data to standard normal margins (and therefore unbounded support) where the transformed density gets estimated by a standard kernel estimator. Then, this estimate is transformed back to uniform margins. Denote $\Phi$, $\Phi^{-1}$, and $\phi$ as the standard Gaussian $cdf$, quantile and density functions. For $\bm s \in \R^2$, let us write short $\bm K(\bm s) = K(s_1)K(s_2)$, and $\bm K_{B_n}(\bm s) = \bm K(B_n^{-1}\bm s) / \det(B_n)$ for some positive definite bandwidth matrix $B_n \in \R^2$. The transformation estimator is defined via
\begin{align}
	\overline c_{j_e,k_e;D_e}(u,v) = \frac{1}{n} \sum_{i=1}^{n} \bm K_{B_n} \begin{pmatrix}
\Phi^{-1}(u) - \Phi^{-1}(U_{j_e|D_e}^{(i)}) \\
\Phi^{-1}(v) - \Phi^{-1}(V_{k_e|D_e}^{(i)})
\end{pmatrix}
/\bigl[\phi\bigl\{\Phi^{-1}(u)\bigr\}\phi\bigl\{\Phi^{-1}(v)\bigr\}\bigr]. \label{def:trafo}
\end{align}
 In order to verify the high-level assumptions \hyperref[Theory:oracle_ass]{\ref*{Theory:oracle_ass}a} and \hyperref[Theory:bound_ass]{\ref*{Theory:bound_ass}a}, we need the following two conditions to hold for all $e \in E_1, \dots, E_{d-1}$:
\begin{enumerate}[label=C\arabic*:,ref=C\arabic*] 
        \item The true pair-copula densities $c_{j_e, k_e;D_e}$ are twice continuously differentiable on $(0, 1)^2$.  \label{C1}
        \item The transformed densities  $\psi_{j_e, k_e;D_e}(x, y) = c_{j_e, k_e;D_e}\bigl\{\Phi(x), \Phi(v)\bigr\}\phi(x)\phi(y)$ have continuous and bounded first- and second-order derivatives on $\R^2$.  \label{C2}
\end{enumerate}
\ref{C1} is a smoothness condition that is very common in nonparametric estimation. \ref{C2} is less standard as it relates to the transformed density. Sufficient conditions for \ref{C2} are given in Lemma A.1 of \citep{Geenens14a} and can be verified for many parametric families, including the ones used in our simulation study.

To avoid unnecessary technicality, we will assume here that the bandwidth matrix is a multiple of the identity matrix: $B_n = b_n \times I_2$. 
\begin{Proposition} \label{prop:chat}
Under conditions \ref{K1}, \ref{K2}, \ref{C1}, and \ref{C2}, the estimator \eqref{def:trafo} satisfies for all $(u,v) \in (0,1)^2$, $e \in E_1, \dots, E_m$,
	\begin{align*}
		\overline c_{j_e, k_e;D_e}(u,v) - c_{j_e, k_e;D_e}(u,v) &= O_p\bigl(b_n^2 + \sqrt{1 / (nb_n^2)}\bigr), \\
		\wh c_{j_e, k_e;D_e}(u,v) - \overline c_{j_e, k_e;D_e}(u, v) &= O_p(a_{e, n}).
     \end{align*}
\end{Proposition}
\begin{proof}
	For the first equality, see Section 3.4 in \citep{Nagler14}. For the second, see \autoref{lem:c} in \autoref{Appendix2}.
\end{proof}\noindent
When the mean-square optimal bandwidth $b_n = O(n^{-1/6})$ is used, the right hand side of the first equality is $O_p\bigl(n^{-1/3}\bigr)$.

\subsection{Estimation of h-functions} \label{BiCopKHE}
Recall that h-functions are actually conditional distribution functions: 
\begin{align*}
	h_{j_e| k_e;D_e}(u| v) = \Pr(U_{j_e|D_e} \le u | U_{k_e|D_e} = v) = \E\bigl\{\ind(U_{j_e|D_e} \le u) | U_{k_e|D_e} = v\bigr\}.
\end{align*}
The second equality relates the conditional $cdf$ to a regression problem. Hence, any nonparametric regression estimator is suitable for estimation of the h-functions. In our case, it is even simpler to integrate the density estimate to obtain an estimate of the corresponding h-function: for the oracle estimators, 
\begin{align} \label{def:hfunc}
	\overline h_{j_e| k_e;D_e}(u| v) := \int_0^u \overline c_{k_e, j_e;D_e}(s, v) ds, \quad \overline h_{k_e| j_e;D_e}(v| u) := \int_0^v \overline c_{j_e, k_e;D_e}(u, s) ds,
\end{align}
and the feasible estimators $\wh h_{j_e| k_e;D_e}$ and $\wh h_{k_e| j_e;D_e}$ are defined similarly.
Such estimators are closely related to the smoothed Nadaraya-Watson estimator of \citep{Hansen04}. In fact, they coincide when we choose diagonal $B_n$ in \eqref{def:trafo}. For an explicit formula, see \eqref{not_h_eq} in \autoref{Appendix2}. The following result puts this estimator in the context of \hyperref[Theory:oracle_ass]{\ref*{Theory:oracle_ass}b} and \hyperref[Theory:bound_ass]{\ref*{Theory:bound_ass}b}.
    \begin{Proposition} \label{prop:hhat}
Under conditions \ref{K1}, \ref{K2}, \ref{C1}, and \ref{C2}, the estimator defined by \eqref{def:hfunc}  and \eqref{def:trafo} satisfies for all $\delta \in (0, 0.5]$,  and $e \in E_1, \dots, E_m$,  
	\begin{align*}
		\sup_{(u, v) \in [\delta, 1- \delta]^2} \bigl\vert \overline h_{j_e| k_e;D_e}(u|v) - h_{j_e| k_e;D_e}(u|v) \bigl\vert  &=  O_{a.s.}\bigl(b_n^2 + \sqrt{\ln n / (nb_n)}\bigr), \\
		\sup_{(u, v) \in [\delta, 1- \delta]^2} \bigl\vert \overline h_{k_e| j_e;D_e}(u|v) - h_{k_e| j_e;D_e}(u|v) \bigl\vert  &=  O_{a.s.}\bigl(b_n^2 + \sqrt{\ln n / (nb_n)}\bigr), \\
       \sup_{(u, v) \in [\delta, 1- \delta]^2} \bigl\vert \wh h_{j_e| k_e;D_e}(u|v) - \overline h_{j_e| k_e;D_e}(u|v) \bigl\vert  &=  O_{a.s.}\bigl(a_{e,n}\bigr), \\
       \sup_{(u, v) \in [\delta, 1- \delta]^2} \bigl\vert \wh h_{k_e| j_e;D_e}(u|v) - \overline h_{k_e| j_e;D_e}(u|v) \bigl\vert  &=  O_{a.s.}\bigl(a_{e,n}\bigr).
	\end{align*}
\end{Proposition}	\noindent
\begin{proof}
See \hyperref[lem:h1]{Lemmas \ref*{lem:h1}} and \ref{lem:h2} in \autoref{Appendix2}.
\end{proof}\noindent
The optimal rate of convergence in the first two equalities is $O_{a.s.}\bigl\{(\ln n / n)^{2/5}\bigr\}$ and attained for $b_n = O\bigl\{(\ln n / n)^{1/5}\bigr\}$. 

\hyperref[Theory:oracle_ass]{Assumption \ref*{Theory:oracle_ass}b} requires that the error of estimating the h-function vanishes faster than the error of pair-copula density estimation. This is readily achieved by using the optimal bandwidth in each component. However, it may be more convenient to use the same bandwidth for pair-copula density as well as h-function estimation. It seems natural to use the optimal rate for pair-copula density estimation, $b_n = O(n^{-1/6})$. But this violates \autoref{Theory:oracle_ass}, because both estimators converge with the same rate: $n^{-1/3}$. To overcome this, we have to increase the speed of $b_n$ by a small amount, i.e., to undersmooth the pair-copula density estimate. When $b_n = \alpha_n n^{-1/6}$, $\alpha_n = o(1)$, the pair-copula density estimators converges with rate $\alpha_n^{-1}n^{-1/3}$ and the h-function estimator with rate $\alpha_n^{2} n^{-1/3} + \alpha_n^{-1/2}n^{-5/12}  = o(\alpha_n^{-1}n^{-1/3})$. But the sequence $\alpha_n$ can converge arbitrarily slow. So we should not expect any problems with using the mean-square optimal rate $b_n = n^{-1/6}$ in practice. This was confirmed by preliminary numerical experiments.

\subsection{Asymptotic normality}

We now put all pieces together and show that the estimator $\wh f_{\mathrm{vine}}$ composed of \eqref{def:fhat},  \eqref{def:trafo}, and \eqref{def:hfunc} is asymptotically normal. We start by establishing the joint asymptotic normality of all components. The proof is deferred to \autoref{Appendix3}. 
\begin{Proposition} \label{prop:an}
	Assume that 
    \begin{enumerate}
	\item  conditions \ref{K1}, \ref{M1}, \ref{C1}, and \ref{C2} hold,
    \item $\wh f_\ell$ and $\wh F_\ell$ are defined by \eqref{def:fhat} with (marginal) bandwidth parameter $b_{n,m}$,
    \item $\wh c_{j_e,k_e;D_e}$ are defined by \eqref{def:trafo} with (copula) bandwidth parameter $b_{n,c}$,
    \item $\wh h_{j_e|k_e;D_e}$ and  $\wh h_{j_e|k_e;D_e}$ are defined by \eqref{def:hfunc} and \eqref{def:trafo} with (h-function) bandwidth parameter $b_{n,h}$,
    \item it holds $b_{n,c} = O(n^{-1/6})$, and for sufficiently large $n$,
    \begin{align*}
		b_{n,c}^2 < b_{n,m} \le b_{n,h} \le \min\{b_{n, c}, n^{-1/6}/\log{n}\}.
\end{align*}
\end{enumerate}
Recall  the definition of $\wh{\bm f}^*(\bm x)$, ${\bm f}^*(\bm x)$, and $d^*$ from \autoref{sec:asdistr}.  It holds for all $\bm x \in R^d$, 
\begin{align} \label{joint_an}
(nb_{n, c}^2)^{1/2}\bigl\{\wh{\bm f}^*(\bm x) - b_{n,c}^2 \bm \mu_{\bm x} - {\bm f}^*(\bm x)\bigr\} &\stackrel{d}{\to} \mathcal{N}_{d^*}\bigl(0, \Sigma_{\bm x}\bigr),
\end{align}
where ${\bm \mu}_{\bm x} = (\bm 0_d^\top, \tilde{\bm \mu}_{\bm x}^\top)^\top$, $\tilde{\bm \mu}_{\bm x} = (\tilde{\mu}_{\bm x, e})_{e \in E_1, \dots, E_{d-1}}$, and $\Sigma_{\bm x}$ is diagonal with first $d$ diagonal entries equal to 0 and remaining  diagonal entries $(\tilde \sigma_{\bm x, e})_{e \in E_1, \dots, E_{d-1}}$. Explicit expressions for $\tilde{\mu}_{\bm x, e}$ and $\tilde \sigma_{\bm x, e}$ are given in \eqref{el_mu} and \eqref{el_sig} in \autoref{Appendix3}.
\end{Proposition} \noindent
The asymptotic normality of $\wh f_{\mathrm{vine}}$ follows by an application of the delta method.
\begin{Corollary}
	Under the assumptions of \autoref{prop:an} it holds for all $\bm x \in \R^d$,
	\begin{align*}
	(nb_{n, c}^2)^{1/2}\bigl\{\wh f_{\mathrm{vine}}(\bm x) -  b_{n,c}^2 \bm \theta^\top \bm \mu_{\bm x}  - f(\bm x)\bigr\} \stackrel{d}{\to} \mathcal{N}\bigl(0, \bm \theta^\top \Sigma_{\bm x} \bm \theta\bigr),
	\end{align*}
	where $\theta_k = \prod_{j \neq k}  f^*_j(\bm x)$, $k = 1, \dots, d^*$, and $\bm \mu_{\bm x}$, $\Sigma_{\bm x}$ are as in \autoref{prop:an}.
\end{Corollary} \noindent

\subsection{Structure selection} \label{Structure}

Finding the optimal structure for vine copulas is extremely difficult. Because of the large number of possibilities, practical approaches are usually based on heuristics. In few situations, expert knowledge can be used to decide which pair-wise dependencies should be modeled explicitly.  If there is no meaningful prior information, the  structure selection algorithm of \citep{Dissmann13} can be adopted. Starting with the first tree, we select the tree that is a maximum (or minimum) spanning tree w.r.t.\ some weight function $w_e$ assigning a weight to each pair of pseudo-observations. The most popular weights are empirical estimates of $\tau_e$,  the (unconditional) Kendall's $\tau$ corresponding to $c_{j_e,k_e;D_e}$. They can be estimated sequentially from the pseudo-observations defined in \autoref{RVineKDE:seqest_alg}. The idea is to choose a structure that captures most of the dependence in lower trees. Other possible weights are the AIC or goodness-of-fit $p$-values corresponding to a pair-copula estimate; see \citep{Czado13} for a discussion. By using kernel density estimators for the pair-copulas, we get a fully nonparametric structure selection algorithm.
	\section{Simulations} \label{Simulations}

In this section, we study the finite sample behavior of a vine copula based kernel density estimator. We illustrate its advantages compared with the classical kernel density estimator in three scenarios that comprise one simplified and two non-simplified target densities. 

\subsection{Implementation of estimators} \label{Simulation:Implementation}

The study was carried out in the statistical computing environment R \citep{R}. We use the implementation of  $\wh f_{\mathrm{vine}}$ introduced in the previous section:
\begin{description}   
	\item {\bfseries Marginal densities} are estimated by the standard kernel density estimator \eqref{def:fhat}. Bandwidths are selected by the plug-in method of \citep{Chacon10}, as implemented in the function \verb|hpi| of the \verb|ks| package \citep{ks}. 
    
    \item {\bfseries Marginal distributions} are estimated by integrating the estimates of the marginal densities.

	\item {\bfseries Pair-copula densities} are estimated by the transformation estimator \eqref{def:trafo} with bandwidth matrix selected by the normal reference rule; see, e.g., Section 3.4 in \citep{Nagler14}.
    
    \item {\bfseries The vine structure} is considered unknown and selected by the method of \citep{Dissmann13} using empirical estimates of $\tau_e$ as weight function (see \autoref{Structure}).
\end{description}
The estimator $\wh f_{\mathrm{vine}}$ is implemented in the R package \texttt{kdevine} \citep{kdevine}. The package also includes estimators for marginals with bounded support as well as more sophisticated pair-copula estimators which further improve the performance. For the classical multivariate kernel density estimator ($\wh f_{\mathrm{mvkde}}$ from here on) we use the function \verb|kde| provided by the \verb|ks| package \citep{ks}. It selects the bandwidths by the plug-in method of \citep{Chacon10}.

\subsection{Performance measurement}

We evaluated the performance of both estimators for three choices of the target density $f$. To gain insight on their convergence behavior under increasing dimension, we consider five different sample sizes $n = 200, 500, 1\,000, 2\, 500, 5\,000$, and three different  dimensions $d=3, 5, 10$. For any fixed target density, sample size, and dimension, we measure the performance as follows:

\begin{enumerate}[label=\arabic*.]
	\item Simulate $n_{sim} = 250$ samples of size $n$, from a $d$-dimensional target density $f$.
	\item On each sample, estimate the density with estimators $\wh f_{\mathrm{vine}}$ and $\wh f_{\mathrm{mvkde}}$.
	\item For each estimator $\wh f \in \bigl\{\wh f_{\mathrm{vine}}, \wh f_{\mathrm{mvkde}}\bigr\}$  and sample, calculate the \emph{integrated absolute error (IAE)} as a performance measure:
	\begin{align*}
	\mathrm{IAE}\bigl(\wh f\bigr) := \int_{\mathds{R}^d}\big\vert \wh f(\bm x) - f(\bm x)\big\vert d\bm x.
	\end{align*}
	The integral is estimated by importance sampling Monte Carlo (e.g., Section 5.2 in \citep{Ripley87}), where we take the true density $f$ as the sampling distribution. The number of Monte Carlo samples was set to $1\, 000$. This gives an unbiased, low-variance estimate of the IAE. 
\end{enumerate}
In the following section we will present the median IAE attained over 250 simulations. Additionally, we use Mood's median test \citep{Gibbons04} to check whether the difference in performance is statistically significant at the 1\% level. Significant results will be indicated by stars above sample size axes of \autoref{Simulations:fig}.

\subsection{Results}

In the following, we illustrate the main insights of our numerical experiments in three examples --- one where the simplifying assumption holds, and two where it does not. Since the simplifying assumption is a property of the copula, we focus on this part and set the marginal densities to standard Gaussian in all scenarios. For these margins, the two estimators $\wh f_{\mathrm{vine}}$ and $\wh f_{\mathrm{mvkde}}$ are asymptotically equivalent when $d=2$. But they become different as soon as the simplifying assumption becomes relevant, i.e., when $d > 2$. Hence, differences in the performance of the two estimators can be directly related to the fact that $\wh f_{\mathrm{vine}}$ assumes a simplified model.   Additional simulation results for common parametric copula families (both simplified and non-simplified) and varying strength of dependence are provided in the online supplement. \todo{Reference}

\subsubsection*{Scenario 1: Gaussian Copula}

\begin{figure}[p]
	\centering 
	\subfloat[Gaussian copula]{\includegraphics[width=\textwidth]{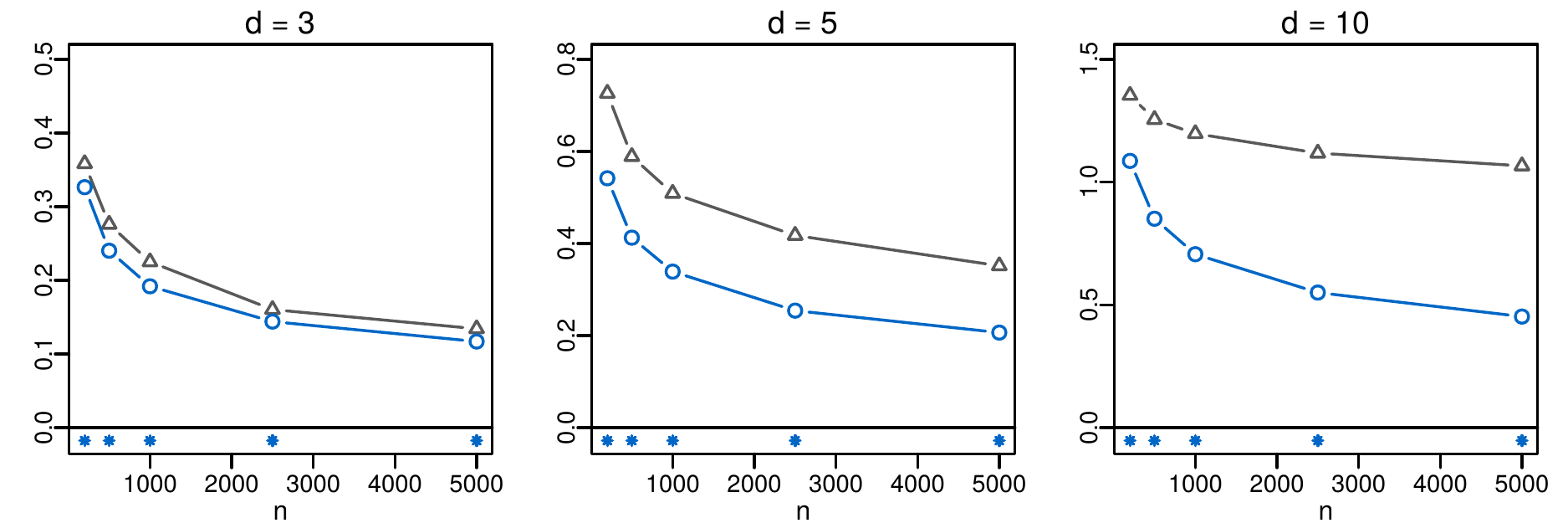} \label{Simulations:1_fig}} \\
 	\subfloat[Gumbel  copula]{\includegraphics[width=\textwidth]{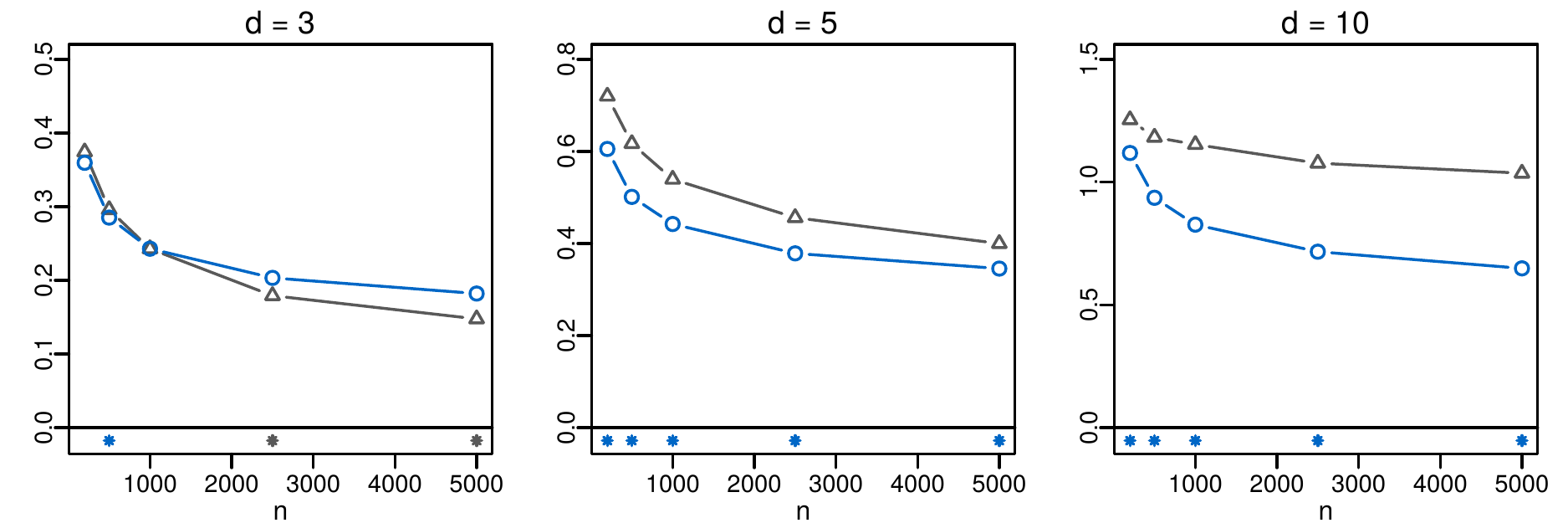} \label{Simulations:2_fig}} \\
    	\subfloat[Non-simplifed Gaussian vine]{\includegraphics[width=\textwidth]{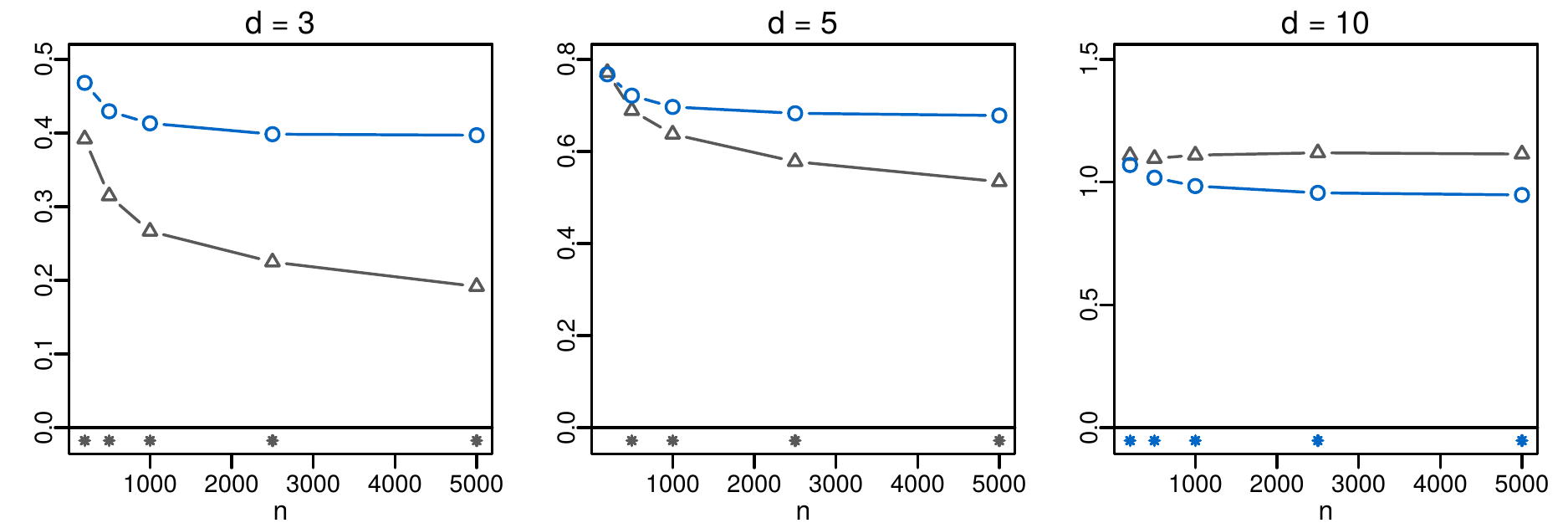} \label{Simulations:3_fig}} \\
	\caption{Median integrated absolute error achieved for varying sample size $n$ and dimension $d$. The estimator $\wh f_{\mathrm{vine}}$ is indicated by circles; $\wh f_{\mathrm{mvkde}}$ by triangles. A star above the sample size means that the corresponding medians were found significantly different at the 1\% level by Mood's median test.} \label{Simulations:fig}
\end{figure}

The first scenario concerns the estimation of a $d$-dimensional Gaussian density. For simplicity, we choose the parameters such that all pair-wise Kendall's $\tau$ equal 0.4 (this corresponds to an association parameter of $\rho \approx 0.6$). Recall that the simplifying assumption is a property of the dependence, i.e.\ the copula. The copula underlying a multivariate Gaussian density is the Gaussian copula which belongs to the class of simplified vine distributions \citep{Stoeber13}. Consequently, the vine copula based estimator is consistent in this situation.

\autoref{Simulations:1_fig} shows the median IAE of $\wh f_{\mathrm{vine}}$ (circles) and $\wh f_{\mathrm{mvkde}}$ (triangles) for varying sample size $n$ and dimension $d$. The vine copula based estimator strictly outperforms the classical estimator by a considerable margin. The difference in IAEs is statistically significant for all dimensions and sample sizes. As predicted by \autoref{Theory:rate_thm}, we observe that --- in contrast to the classical kernel density estimator --- the vine copula based estimator converges at the same rate independent of dimension. Thus, the gap widens as dimension or sample size increase. For $d=5$, $\wh f_{\mathrm{vine}}$ is almost two times as accurate; for $d=10$ almost three times as accurate.  These numbers are remarkable considering how slowly $\wh f_{\mathrm{mvkde}}$ can improve its accuracy when increasing sample size. The same conclusions can be drawn from the additional simulation results for simplified models provided in the online supplement. \todo{Reference}

\subsubsection*{Scenario 2: Gumbel copula}

Our second scenario, a Gumbel copula coupled with standard normal margins, violates the simplifying assumption; see Theorem 3.1 in \citep{Stoeber13}. Again, we choose the parameter of the Gumbel copula such that all pair-wise Kendall's $\tau$ equal 0.4 (this corresponds to a Gumbel copula parameter $\theta \approx 1.67$). In this case, $\wh f_{\mathrm{mvkde}}$ is guaranteed to outperform $\wh f_{\mathrm{vine}}$ as $n \to \infty$, because the latter is not consistent. On finite samples, however, the picture seems to be different.

The performance of the two estimators in this scenario is displayed in \autoref{Simulations:2_fig}. For $d=3$, $\wh f_{\mathrm{vine}}$ is slightly worse than its competitor, but the difference is only significant for large sample sizes. For increasing dimension, the gap widens in favor of $\wh f_{\mathrm{vine}}$ which performs significantly better for $d =5$ and $d= 10$. For $d=10$ and $n = 5\,000$, the vine copula based estimator is almost two times as accurate --- although it is not consistent. Since $\wh f_{\mathrm{mvkde}}$ converges so slowly, an extremely large number of observations would be required until it becomes the better choice. But for commonly available sample sizes and $d>3$, the vine copula based estimator is preferable. The same findings hold for the additional simulation results for non-simplified models provided in the online supplement.

\subsubsection*{Scenario 3: Non-simplified Gaussian vine}

Lastly, we want to investigate how the vine copula based estimator behaves in a sort of `worst case scenario'. We set up a non-simplified vine copula with Gaussian pair-copulas and formulate their parameters as a regression on the conditioning variables implied by the vine. For each conditional pair-copula, the correlation parameter function $\rho_e\colon [0,1]^{|D_e|} \to [-1, 1]$ describes a linear hyperplane ranging from $-1$ to $1$:
\begin{align*}
	\rho_e(\bm u_{D_e}) = 1 - \frac 2 {\vert D_e \vert} \sum_{j \in D_e} u_{j}, \qquad \mbox{for } e \in E_m,\, m \ge 2.
\end{align*}
Since $\int \rho_e(\bm u_{D_e}) d\bm u_{D_e} = 0$ for all $e \in E_2, \dots, E_{d-1}$, we also set $\rho_e \equiv 0$ for $e \in E_1$. This model is severely violating the simplifying assumption for each conditional pair in the vine.

The results for this scenario are shown in \autoref{Simulations:3_fig}. The vine copula based estimator performs significantly worse for $d=3, 5$. Remarkably, $\wh f_{\mathrm{vine}}$ manages to significantly outperform the classical estimator for $d=10$. The severely non-simplified dependence structure appears to be too difficult to identify even for a nonparametric estimator that does not rely on the simplifying assumption. Extrapolating the curves, we can expect that to hold for sample sizes much larger than those considered in our study.  Also, we can expect the advantage of $\wh f_{\mathrm{vine}}$ to become even bigger in higher dimensions. We can conclude that even in this extremely unfavorable example, the estimator $\wh f_{\mathrm{vine}}$ proves useful when more than a few variables are involved.

	\section{Application} \label{Application}

 We revisit a classification problem from astrophysics which has previously been investigated in \citep{Bock04}.
 In their study, the authors consider synthetic data imitating measurements taken on images from the \emph{MAGIC (Major Atmospheric Gamma-ray Imaging Cherenkov) Telescopes} located on the Canary islands. The goal is to identify primary gamma rays (the signal) amongst a large amount of hadron showers (background noise). The authors of the study evaluate the performance of several classification methods and judge the kernel density based Bayes classifier as one of the most convincing. We aim to augment their results and investigate how the vine copula based kernel density estimator performs on this problem.
 
 The data set is available from the \emph{UCI Machine Learning Repository} web page (url: \url{https://archive.ics.uci.edu/ml/datasets/MAGIC+Gamma+Telescope}) and consists of $n=19\,020$ observations on $d=10$ variables. $n_G = 12\,332$ of the observations are classified as gamma (signal) and $n_H=6\,688$ as hadron (background). For more information on the astrophysical background and a more thorough description of the data we refer the reader  to \citep{Bock04} and the UCI web page.
 
 Bayes classifiers follow the idea of maximizing the posterior probability of a class given the data. Let $G$ (for gamma) and $H$ (for hadron) be the two classes and $\wh f_G$ and $\wh f_H$ be two estimates fitted separately in each class. Assume further we have knowledge of the class prior probabilities $\pi_G, \pi_H$. With a straightforward application of Bayes' theorem, we can estimate the posterior probability that the class is $G$ as
 \begin{align}
 \wh \Prob(\mbox{Class} = G | \bm X  = \bm x) = \frac{\pi_G \wh f_G(\bm x)}{\pi_G \wh f_G(\bm x) + \pi_H \wh f_H(\bm x)}, \label{Application:Bayes_eq}
 \end{align}
 where $\bm x$ is a realization of the random vector $\bm X$. In the most general case, we classify an observation as $G$ whenever the estimated posterior probability is greater than $\alpha = 0.5$. However, by changing the threshold $\alpha$ we can furthermore  control how many observations get classified as $G$, and thereby influence key quantities such as the \emph{false positive rate (FPR)} or \emph{true positive rate (TPR)}. The FPR is defined as the ratio of the number of false positives (here: hadron events that were misclassified as gamma) and the number of all negative (hadron) events. The TPR is defined as the ratio of the number of correctly classified positive (gamma) events and the number of all positive events. In general, it is desirable to have a low FPR and a high TPR. But usually, there is a tradeoff between the two quantities: If we increase the threshold level $\alpha$, a higher posterior probability is required for an observation to get classified as gamma event. As a result, less observations will be classified as gamma event, which in turn reduces \emph{both} FPR and TPR.

 \begin{figure}[t]
 	\centering
 	\includegraphics[width=0.85\linewidth]{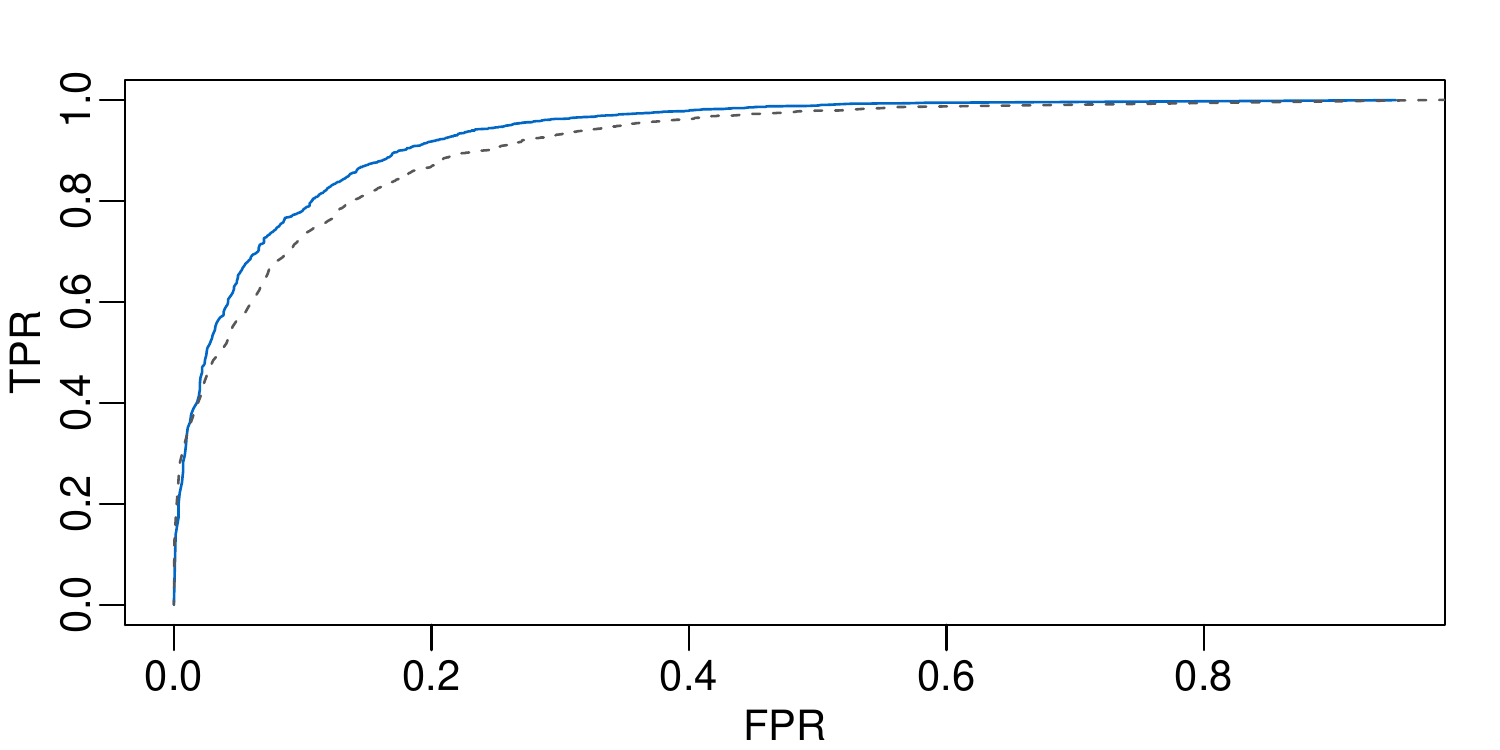}
 	\caption{ROC curves for Bayes classifiers based on the vine copula based estimator (solid line) and classical multivariate kernel density estimator (dashed line).}
 	\label{Application:ROC_fig}
 \end{figure}
 
We repeat the experiment of \citep{Bock04} with the vine copula based and classical kernel estimators. The implementations are similar to our simulation study (see \autoref{Simulation:Implementation}). As is common in applications, we induce sparsity of the estimated model by adding an independence test to the structure selection algorithm; see Section 4 in \citep{Dissmann13}. We also found it necessary to multiply the marginal bandwidth parameters of $\wh f_{\mathrm{vine}}$ by 2 to stabilize the classification boundary in low-density regions. The experiment's setup is the following: First, the densities for each class are estimated on the first 2/3 of the data which is used as training set. These estimates are used in combination with \eqref{Application:Bayes_eq} to obtain class predictions for the remaining 1/3. For simplicity, the prior probabilities are set to $\pi_G=\pi_H=0.5$. The predictions are then compared to the actual class of the observations which allows to asses the quality of the predictions.

\autoref{Application:ROC_fig} shows the \emph{receiver operating characteristic (ROC) curve} which displays the TPR as a function of the FPR.  It was noted in \citep{Bock04} that in this application the focus is on low FPR level; in particular the 0.01, 0.02, 0.05, 0.1 and 0.2 levels. The TPR values of the ROC curves at these levels are additionally displayed in \autoref{Application:TPR_tab}. 
The ROC curve of the vine copula based estimator lies above the curve of the classical multivariate kernel density estimator almost everywhere. This means that for a target FPR level, the vine copula based classifier is able to identify more observations correctly as signal events than the classical multivariate kernel density estimator. The results confirm what we could expect from our simulation study where, for $d=10$ and several thousand observations, the vine copula based approach delivered much more accurate estimates.

\begin{table}[t]
	\centering 
	\begin{tabular}{r|ccccc}
		FPR & 0.01 & 0.02 & 0.05 & 0.1 & 0.2 \\ \hline 
		vine &  0.335 & 0.428 & 0.652 & 0.780 & 0.918 \\
        mvkde & 0.335 & 0.408 & 0.567 & 0.730 & 0.868
	\end{tabular}
	\caption{True positive rates for the two estimators (second and third row) for given target levels of the false positive rate (first row).}
	\label{Application:TPR_tab}
\end{table}

But also in comparison with other classification algorithms, the classifier based on $\wh f_{\mathrm{vine}}$ performs extraordinary well. A total of 14 algorithms were surveyed in \citep{Bock04}, including variants of classification trees and neural networks, as well as the popular nearest-neighbor method and support vector machine. Two of the main performance measures used in their study are the average of the TPR at the 0.01, 0.02 and 0.05 FPR levels (termed \emph{loacc}), and the average of the TPR at the 0.1 and 0.2 FPR levels (termed \emph{highacc}). From \autoref{Application:TPR_tab} we calculate $loacc = 0.472$ and $highacc = 0.849$. None of the 14 algorithms was able to produce a better \emph{loacc} value than our approach. And only one method, random forests, delivered a slightly higher \emph{highacc} of $0.852$. This is particularly remarkable when we consider that the parameterization of our estimator was not tuned with respect to classification accuracy (unlike other classification algorithms). It might well be that the performance can be further improved by bandwidth and structure selection strategies that aim for classification rather than estimation accuracy.
	\section{Further discussion} \label{Conclusion}

In this paper, we discuss a vine copula approach to nonparametric density estimation. By assuming that the target density belongs to the class of simplified vine densities, we can divide the estimation of a $d$-dimensional density into several one- and two-dimensional tasks. This allows us to achieve faster  convergence rates than classical nonparameteric estimators when $d>3$. In particular, the speed of convergence is independent of dimension. The advantages of this approach become more and more striking as dimension increases. It shows that a simplified  vine model for the dependence between variables is an appealing structure for nonparametric problems. For example, we can expect that similar results can be obtained for copula-based regression models \citep{Noh13, Kraus15}.

The crunchpoint in our approach is the simplifying assumption. If the simplifying assumption is not satisfied, the proposed estimator is not consistent --- but can nevertheless outperform its competitor in most practicable situations. However, the latter finding may not be true if the simplifying assumption is violated in an extreme fashion and dimension is small. We guess that this is a rather unlikely situation to encounter in real data. However, appropriate tests for a formal empirical assessment have yet to be developed. From a theoretical point of view, this answer is highly unsatisfying and several urging questions arise:
\begin{itemize}
    \item How dense does the set of simplified densities lie in the set of all densities? Put differently: how far off can we be by assuming a simplified model?
	\item How can we interpret the components of an estimated simplified model when the assumption does not hold?
\end{itemize}

Owing to the infancy of vine copula models, these questions remain open to this day. But several recent works have advanced the understanding of the simplifying assumption.  A discussion of its appropriateness can be found in \citep{Haff10}. Copula classes where the simplifying assumption is satisfied are given in \citep{Stoeber13}. In \citep{Gijbels15}, a general estimator of the copula was proposed for the case where a covariate affects only the marginal distributions (i.e., when the simplifying assumption does hold). Semiparametric estimation of three-dimensional non-simplified PCCs was tackled in \citep{Acar12};  a test for the simplifying assumption was proposed in \citep{Acar13} under a semiparametric model. The empirical pair-copula, an extension of the empirical copula to simplified vine copulas, was analyzed in \citep{Haff15}. The authors conjecture that this estimator converges at the parametric rate --- even when pseudo-observations are used. The situation is different from ours since empirical copulas do not suffer the curse of dimensionality.

The notion of \emph{partial vine copula approximations (PVCA)}, i.e., the limit of a step-wise estimator under a simplified model, was introduced in \citep{Spanhel15}. The authors show that the PVCA is not necessarily the best simplified approximation to the true density. They further illustrate in an example that spurious dependence patterns can appear in trees $T_m, m \ge 3,$ when the simplifying assumptions has falsely been assumed in previous trees $T_{m'}, 2 \le m' \le m$. This property may not matter much in terms of estimation accuracy, but can corrupt the interpretability of an estimated PVCA. The estimator proposed in this paper is in fact an estimator of the PVCA. Our results suggest that the PVCA is a useful inferential object in any case:
\begin{itemize}
	\item Any $d$-dimensional PVCA can be consistently estimated at a rate that is equivalent to a two-dimensional problem.
	\item If the simplifying assumption does hold, the PVCA coincides with the true density.
    \item If the simplifying assumption does not hold, inference of the PVCA is still less difficult than inference of the actual density. This led to the following observation: On finite samples, a consistent estimate of the PVCA can be much closer to the true density than a consistent estimate of density itself (see Scenario 2 in \autoref{Simulations}).
\end{itemize}

A related perspective on the phenomenon is that the simplifying assumption allows us to achieve more accurate estimates by model shrinkage. We incorporate the additional  `information' that the simplifying assumption is at least approximately true. This allows us to reduce the set of possible solutions and thereby make the estimation problem `less difficult'. The most well known example of a shrinkage estimator is the sample variance. When dividing by $n$ instead of $n-1$ we give up unbiasedness of the estimator in order to achieve a smaller error. The same holds true for the vine copula based density estimator: if we make the simplifying assumption although it is not satisfied, we introduce additional bias. In fact, we even give up consistency of the estimator in order to achieve better finite-sample accuracy.

The main advantage of the vine copula based approach is striking: Classical multivariate nonparametric density estimators converge very slowly to the true density when more than a few variables enter the model. Hence, one was unable to benefit from the increasing number of observations in modern data. A vine copula based estimator, on the other hand, converges at a high speed, no matter how many variables are involved. This makes it particularly appealing in the age of big data.

\subsection*{Acknowledgements}
The first author acknowledges financial support by a research stipend of the Technische Universit{\"a}t M{\"u}nchen. The second author is supported by the German Research Foundation (DFG grant CZ 86/4-1). The authors thank two anonymous referees whose remarks led to a considerably improved contribution.

	\appendix

\section{Proof of Theorem 1} \label{Appendix}

\setlength{\parskip}{0pt}

The proof consists of three steps. In the first step, we show by induction that all pseudo-observations converge sufficiently fast to the true observations. In the second step, we establish pointwise consistency of the feasible pair-copula density estimators $\wh c_{j_e,k_e;D_e}$ and conditional distribution function estimators $\wh F_{j_e|D_e}$ and $\wh F_{k_e|D_e}$. In the last step, we combine these results to establish the consistency of $\wh f_{\mathrm{vine}}$.


\subsubsection*{Step 1: Convergence of pseudo-observations}

We will show by induction that for all $e \in E_1, \dots, E_{d-1}$, $i = 1, \dots, n$,
\begin{align} \label{Theory:Uproof_eq}
	\begin{aligned}
		  \wh U_{j_e|D_e}^{(i)} - U_{j_e|D_e}^{(i)} = o_{a.s.}(n^{-r}) , \quad  \wh U_{k_e|D_e}^{(i)} - U_{k_e|D_e}^{(i)}  = o_{a.s.}(n^{-r}).
	\end{aligned} 
\end{align}
Let $e \in E_1$ (the conditioning set $D_e$ is empty). Because of \hyperref[Theory:f_ass]{\ref*{Theory:f_ass}b} we have,
\begin{align*} 
 \bigl\vert \wh U_{j_e}^{(i)} - U_{j_e}^{(i)} \bigr\vert &=  \bigl\vert \wh F\bigl(X_{j_e}\bigr) - F\bigl(X_{j_e}\bigr) \bigl\vert \le \sup_{x_{j_e} \in \Omega_{X_{j_e}}}  \bigl\vert \wh F(x_{j_e}) - F(x_{j_e}) \bigl\vert  = o_{a.s.}(n^{-r}),
\end{align*}
and the same argument applies to the second equality of  \eqref{Theory:Uproof_eq}.
Now consider $e \in E_m$, $1 \le m \le d-2$, and assume that \eqref{Theory:Uproof_eq} holds for all $e \in E_{m}$. Recall that all pseudo-observations for $e^\prime \in E_{m+1}$ can be written as $\wh U_{j_e | D_e \cup k_e}^{(i)}$ or $\wh U_{k_e | D_e \cup j_e}^{(i)}$ for some $e \in E_m$. By the definition of the pseudo-observations and the triangle inequality,
\begin{eqnarray*}
	\bigl\vert \wh U_{j_e | D_e \cup k_e}^{(i)}  - U_{j_e | D_e \cup k_e}^{(i)}  \bigr\vert 
    &=& \bigl\vert \wh h_{j_e | k_e; D_e}\bigl\{\wh U_{j_e | D_e}^{(i)}  | \wh U_{k_e | D_e}^{(i)} \bigr\} -  h_{j_e | k_e; D_e}\bigl\{U_{j_e | D_e}^{(i)}  | U_{k_e | D_e}^{(i)} \bigr)\}\bigr\vert \\
     &\le & \phantom{+ \;} \bigl\vert \wh h_{j_e | k_e; D_e}\bigl\{\wh U_{j_e | D_e}^{(i)}  | \wh U_{k_e | D_e}^{(i)} \bigr\} -  \overline h_{j_e | k_e; D_e}\bigl\{\wh U_{j_e | D_e}^{(i)}  | \wh U_{k_e | D_e}^{(i)} \bigr\}  \bigr\vert \\
    & &+ \; \bigl\vert \overline h_{j_e | k_e; D_e}\bigl\{\wh U_{j_e | D_e}^{(i)}  | \wh U_{k_e | D_e}^{(i)} \bigr\} -  h_{j_e | k_e; D_e}\bigl\{\wh U_{j_e | D_e}^{(i)}  | \wh U_{k_e | D_e}^{(i)} \bigr\}  \bigr\vert \\
       & & + \; \bigl\vert h_{j_e | k_e; D_e}\bigl\{\wh U_{j_e | D_e}^{(i)}  | \wh U_{k_e | D_e}^{(i)} \bigr\} -   h_{j_e | k_e; D_e}\bigl\{U_{j_e | D_e}^{(i)}  | U_{k_e | D_e}^{(i)} \bigr\}  \bigr\vert \\
       &=&  H_{1,n} + H_{2, n} + H_{3, n}
\end{eqnarray*}
Note that, almost surely, each realization of $(U_{j_e | D_e}^{(i)},  U_{k_e | D_e}^{(i)})$ is  contained in $[\delta_i, 1 - \delta_i]^2$ for $\delta_i := \min\bigl\{U_{j_e | D_e}^{(i)}, U_{k_e | D_e}^{(i)}, 1 - U_{j_e | D_e}^{(i)}, 1 - U_{k_e | D_e}^{(i)} \bigr\} >0$. And by invoking \eqref{Theory:Uproof_eq} we see that for sufficiently large $n$, also each realization of $(\wh U_{j_e | D_e}^{(i)},  \wh U_{k_e | D_e}^{(i)})$ is contained in $[\delta_i/2, 1 - \delta_i/2]^2$. Together with  \hyperref[Theory:oracle_ass]{\ref*{Theory:oracle_ass}b} and \hyperref[Theory:bound_ass]{\ref*{Theory:bound_ass}b} this yields for large $n$,
\begin{align*}
	H_{1, n} &\le \sup_{(u, v) \in [\delta_i/2, 1-\delta_i/2]^2} \bigl\vert \wh h_{j_e|k_e;D_e}(u|v) - \overline h_{j_e|k_e;D_e}(u|v) \bigr\vert = O_{a.s.}(a_{e,n}), \\
	H_{2, n} &\le \sup_{(u, v) \in [\delta_i/2, 1-\delta_i/2]^2} \bigl\vert \overline h_{j_e|k_e;D_e}(u|v) - h_{j_e|k_e;D_e}(u|v) \bigr\vert = o_{a.s.}(n^{-r}),
\end{align*}
and invoking \eqref{Theory:Uproof_eq},
\begin{align*}
a_{e,n} = \sup_{i = 1, \dots, n} \vert\wh U_{j_e|D_e}^{(i)} - U_{j_e|D_e}^{(i)}\bigr\vert +  \bigl\vert\wh U_{k_e|D_e}^{(i)} - U_{k_e|D_e}^{(i)}\bigr\vert = o_{a.s.}(n^{-r}),
\end{align*} 
which gives $H_{1,n} = o_{a.s.}(n^{-r})$. It remains to show that $H_{3,n} = o_{a.s.}(n^{-r})$. Let $\nabla h_{j_e | k_e; D_e}$ denote the gradient of $h_{j_e | k_e; D_e}$. A first-order Taylor approximation of $h_{j_e | k_e; D_e}\bigl(\wh U_{j_e | D_e}^{(i)}  | \wh U_{k_e | D_e}^{(i)} \bigr)$ around $\bigl(U_{j_e | D_e}^{(i)}, U_{k_e | D_e}^{(i)}\bigr)$ yields
\begin{align*}
	H_{3, n} &\le \biggl\vert \nabla^\top h_{j_e | k_e; D_e}\bigl(U_{j_e | D_e}^{(i)}  | U_{k_e | D_e}^{(i)}\bigr)  \begin{pmatrix}
	\wh U_{j_e | D_e}^{(i)} -  U_{j_e | D_e}^{(i)} \\
     \wh U_{k_e | D_e}^{(i)} -  U_{k_e | D_e}^{(i)}
\end{pmatrix} \biggr\vert  + o_{a.s.}\begin{pmatrix}
	\wh U_{j_e | D_e}^{(i)} -  U_{j_e | D_e}^{(i)} \\
     \wh U_{k_e | D_e}^{(i)} -  U_{k_e | D_e}^{(i)}
\end{pmatrix}.
 \end{align*}
 	Invoking \eqref{Theory:Uproof_eq}, we get $H_{3, n} = o_{a.s.}(n^{-r})$. This establishes the first equality of \eqref{Theory:Uproof_eq} for all $e \in E_{m+1}$. The second equality follows by symmetric arguments and the induction is complete.
\subsubsection*{Step 2: Consistency of conditional \emph{cdf} and pair-copula density estimators}

With arguments almost identical to those in Step 1, we can furthermore show that for all $e \in E_1, \dots, E_{d-1}$, and all $\bm x \in \Omega_{\bm X}$,
\begin{align} 
	\begin{aligned}\label{Theory:Fproof_eq}
		 \wh F_{j_e|D_e}\bigl(x_{j_e} | {\bm x}_{D_e}\bigr) - F_{j_e|D_e}\bigl(x_{j_e} | \bm x_{D_e}\bigr) &= o_p(n^{-r}), \\
		\wh F_{k_e|D_e}\bigl(x_{k_e} |{\bm x}_{D_e}\bigr)- F_{k_e|D_e}\bigl(x_{k_e} | \bm x_{D_e}\bigr) &= o_p(n^{-r}).
	\end{aligned}
\end{align}
Next, we establish that for all $e \in E_1, \dots, E_{d-1}$, and all $(u, v) \in (0, 1)^2$,
\begin{align}  \label{Theory:cproof_eq}
 \wh c_{j_e,k_e;D_e}\bigl(u, v\bigr) -  c_{j_e,k_e;D_e}\bigl(u, v\bigr) = O_p(n^{-r}).
\end{align}
The triangle inequality gives
\begin{eqnarray*}
    & & \bigl\vert \wh c_{j_e,k_e;D_e}\bigl(u, v\bigr) -  c_{j_e,k_e;D_e}\bigl(u, v\bigr)\bigr\vert  \\
    &\le & \bigl\vert \wh c_{j_e,k_e;D_e}\bigl(u, v\bigr) - \overline c_{j_e,k_e;D_e}\bigl(u, v\bigr)\bigr\vert + \bigl\vert \overline c_{j_e,k_e;D_e}\bigl(u, v\bigr) -  c_{j_e,k_e;D_e}\bigl(u, v\bigr)\bigr\vert  \\
       &=& R_{n,1} + R_{n, 2}.
\end{eqnarray*}
We have $R_{n, 1} = o_{a.s.}(n^{-r})$ by \hyperref[Theory:bound_ass]{\ref*{Theory:bound_ass}a} and \eqref{Theory:Uproof_eq}, whereas $R_{n, 2} = O_p(n^{-r})$ by  \hyperref[Theory:oracle_ass]{\ref*{Theory:oracle_ass}a}.

\subsubsection*{Step 3: Consistency of the vine copula based density estimator}

The consistency of $\wh f_{\mathrm{vine}}$ now follows from \eqref{Theory:cproof_eq} and \hyperref[Theory:f_ass]{\ref*{Theory:f_ass}a} (second equality) together with \eqref{Theory:Fproof_eq} and the fact that $c_{j_e, k_e; D_e}$ is continuously differentiable (third equality):
\begin{align*} \allowdisplaybreaks[1]
	\wh f_{\mathrm{vine}}(\bm x) &= \prod_{k=1}^{d-1} \prod_{e \in E_k} \wh c_{j_e, k_e; D_e} \bigl\{\wh F_{j_e|D_e}(x_{j_e}|{\bm x}_{D_e}), \, \wh F_{k_e|D_e}(x_{k_e}|{\bm  x}_{D_e})\bigr\} \times \prod_{j=1}^d \wh f_j(x_j)\\ 
	&= \prod_{k=1}^{d-1} \prod_{e \in E_k}  \biggl[c_{j_e, k_e; D_e} \bigl\{\wh F_{j_e|D_e}(x_{j_e}|{\bm x}_{D_e}), \,  \wh F_{k_e|D_e}(x_{k_e}|{\bm  x}_{D_e})\bigr\} + O_p(n^{-r})\biggr] \\ 
& \phantom{=} \times \prod_{j=1}^d \bigl\{f_j(x_j) +O_p(r^{-r})\bigr\}\\
	&= \prod_{k=1}^{d-1} \prod_{e \in E_k}  \biggl[c_{j_e, k_e; D_e} \bigl\{ F_{j_e|D_e}(x_{j_e}|{\bm x}_{D_e}), \,  F_{k_e|D_e}(x_{k_e}|{\bm  x}_{D_e})\bigr\} + O_p(n^{-r}) + o_p(n^{-r})\biggr] \\ 
	& \phantom{=} \times \prod_{j=1}^d \bigl\{f_j(x_j) +O_p(n^{-r})\bigr\}\\\ 
	&=  f(\bm x)  +  O_p(n^{-r}). \tag*{\qed}
\end{align*}


\section{Lemmas} \label{Appendix2}


\subsection{Notation} \label{sec:not}
To ease notation in the following proofs, we write 
$(u,v) = (w_1, w_2) = \bigr(\Phi(z_1), \Phi(z_2)\bigl)$,
\begin{align} \label{not_obs_eq}
	W_1^{(i)} := U_{j_e|D_e}^{(i)}, \;   W_2^{(i)} := U_{k_e|D_e}^{(i)}, \;   Z_1^{(i)} := \Phi^{-1}\bigl(U_{j_e|D_e}^{(i)}\bigr), \;   Z_2^{(i)} := \Phi^{-1}\bigl(U_{k_e|D_e}^{(i)}\bigr).
\end{align}
In this notation, the (oracle) transformation pair-copula density estimator is
\begin{align*} 
	\overline c(u,v) = \overline c\bigl\{\Phi(z_1), \Phi(z_2)\bigr\} = \frac{1}{n} \sum_{i=1}^{n} \frac{K_{b_n}\bigl(z_1 - Z_1^{(i)}\bigr) K_{b_n}\bigl(z_2 - Z_2^{(i)}\bigr)}{\phi(z_1)\phi(z_2)}. 
	\end{align*}
    The corresponding (oracle) h-function estimator $\overline h$ is obtained by integration of $\overline c$:
	\begin{align} \label{not_h_eq}
	\overline h(u |v) = \overline h\bigl\{\Phi(z_1) | \Phi(z_2)\bigr\} &= \frac{1}{n} \sum_{i=1}^{n} \frac{J_{b_n}\bigl(z_1 - Z_1^{(i)}\bigr) K_{b_n}\bigl(z_2- Z_2^{(i)}\bigr)}{\phi(z_2)}, 
	\end{align}
    where $J_{b_n}(\cdot) = \int_{-\infty}^\cdot K_{b_n}(s) ds$. The feasible estimators $\wh c$ and $\wh h$ are obtained by replacing $W_j^{(i)}$ and $Z_j^{(i)}$ with pseudo-observations $\wh W_j^{(i)}$ and $\wh Z_j^{(i)} := \Phi^{-1}(\wh W_j^{(i)})$. Finally, we write
        \begin{align*} 
    	 a_{n} = \sup_{i \in \{1, \dots, n\}} \bigl\vert\wh W_1^{(i)} - W_1^{(i)}\bigr\vert + \sup_{i \in \{1, \dots, n\}} \bigl\vert\wh W_2^{(i)} - W_2^{(i)}\bigr\vert.
	\end{align*}
    
\subsection{Results}
\
\begin{Lemma} \label{lem:c}
Under conditions \ref{K1}, \ref{K2}, \ref{C1},  and \ref{C2} it holds for all $(u,v) \in (0,1)^2$, 
	\begin{align*}
		\wh c(u,v) = \overline c(u, v) + O_{a.s}(a_n). 
        \end{align*}
\end{Lemma}
\begin{proof}
By a first-order Taylor approximation of $\Phi^{-1}$, $j =1, 2$,
\begin{align} 
\begin{aligned}  \label{Z_eq}
	 \wh Z_j^{(i)} - Z_j^{(i)} &= (\wh W_j^{(i)} - W_j^{(i)})/ \phi(Z_j^{(i)}) + o_{a.s.}(\wh W_j^{(i)} - W_j^{(i)}) \\
     &= 1/\phi(Z_j^{(i)}) \times O_{a.s.}(a_n), 
     \end{aligned}
\end{align}
where the $O_{a.s.}(a_n)$ term does not depend on the index $i$ since the supremum was taken.
Denote $\nabla_{\bm z} = (\partial/\partial z_1, \partial/\partial z_2)^\top$. A first-order Taylor approximation of $K$ yields
	\begin{align*}
     &\phantom{=} \phi(z_1)\phi(z_2) \bigl\vert\wh c \bigl\{\Phi(z_1), \Phi(z_2) \bigr\}- \overline c\bigl\{\Phi(z_1), \Phi(z_2) \bigr\}\bigr\vert \\
  &=  \biggl\vert \frac{1}{n} \sum_{i= 1}^n K_{b_n}\bigl(z_1 - \wh Z_1^{(i)}\bigr) K_{b_n}\bigl(z_2 - \wh Z_2^{(i)}\bigr) -  \frac{1}{n} \sum_{i= 1}^n K_{b_n}\bigl(z_1 - Z_1^{(i)}\bigr) K_{b_n}\bigl(z_2 - Z_2^{(i)}\bigr)  \biggr\vert\\
  &=  \biggl\vert \frac{1}{n} \sum_{i= 1}^n \nabla_{\bm z} \bigl\{K_{b_n}(z_1 -Z_1^{(i)}) K_{b_n}(z_2 - Z_2^{(i)})\bigr\} \begin{pmatrix}
\wh Z_1^{(i)} - Z_1^{(i)} \\
\wh Z_2^{(i)} - Z_2^{(i)}
\end{pmatrix} + o_{a.s.}\biggl\{\begin{pmatrix}
\wh Z_1^{(i)} - Z_1^{(i)} \\
\wh Z_2^{(i)} - Z_2^{(i)}
\end{pmatrix}\biggr\}\biggr\vert \\
&\le  \biggl\vert \frac{1}{n} \sum_{i= 1}^n \nabla_{\bm z} \bigl\{K_{b_n}(z_1 - Z_1^{(i)}) K_{b_n}(z_2 - Z_2^{(i)})\bigr\} \begin{pmatrix}
1 / \phi\bigl(Z_1^{(i)}\bigr) \\
1 / \phi\bigl(Z_2^{(i)}\bigr)
\end{pmatrix} \biggr\vert \times O_{a.s.}(a_n),
	\end{align*}
    where the last inequality is due to \eqref{Z_eq}.  Since  $K_{b_n}$ is zero outside of $[-b_n, b_n]$, we can bound this further by
    \begin{align} \label{displ}
 \eta_n(\bm z) \times \biggl\vert \nabla_{\bm z} \biggl\{\frac{1}{n} \sum_{i= 1}^n K_{b_n}(z_1 - Z_1^{(i)}) K_{b_n}(z_2 - Z_2^{(i)})\biggr\}  \biggr\vert \times O_{a.s.}(a_n),
\end{align}
where $\eta_n(\bm z) := \sup_{y \in [\min\{z_1,z_2\} - b_n, \max\{z_1, z_2\} + b_n]} 1 / \phi(y) = O(1)$ for all $\bm z \in \R^2$. The second term is the absolute value of the gradient of a classical kernel density estimator. Since the derivatives of $\psi$ are continuous and bounded by \ref{C2}, it holds,
    \begin{align*}
  \biggl\vert \nabla_{\bm z} \biggl\{\frac{1}{n} \sum_{i= 1}^n K_{b_n}(z_1 - Z_1^{(i)}) K_{b_n}(z_2 - Z_2^{(i)})\biggr\}  \biggr\vert = \bigl\vert \nabla_{\bm z} \psi(z_1,z_2)\bigr\vert + o_{a.s.}(1),
\end{align*}
see Theorem 9 in \citep{Hansen08}. Plugging this into \eqref{displ} proves our claim. 
\end{proof}

\begin{Lemma} \label{lem:h1}
Under conditions \ref{K1}, \ref{K2}, \ref{C1}, and \ref{C2}  it holds for all $(u,v) \in (0,1)^2$, $\delta \in (0, 0.5]$,  

	\begin{align*}
		\sup_{(u, v) \in [\delta, 1- \delta]^2} \bigl\vert \overline h(u|v) - h(u|v) \bigl\vert  =  O_{a.s.}\bigl(b_n^2 + \sqrt{\ln n / (nb_n)}\bigr).
	\end{align*}
\end{Lemma}	\noindent
 \begin{proof}
Equations 40 and 41 in \citep{Hansen04} yield
\begin{align*}
	\E\bigl\{\overline h(u | v)\bigr\} - h\bigr(u | v \bigl) = b_n^2 \beta(u, v) + o(b_n^2),
\end{align*}
for some bias term $\beta(u, v)$ involving $h$ and $\phi$ as well as their first- and second order derivatives. Since all parts are continuous on $[\delta, 1- \delta]^2$ by \ref{C1} for all $\delta \in (0, 0.5]$, it holds
\begin{align*}
	\sup_{(u, v) \in [\delta, 1 - \delta]^2} \bigl\vert \E\bigl\{\overline h(u | v)\bigr\} - h\bigr(u | v)  \bigr\vert = O_{a.s.}\bigl(b_n^2\bigr).
\end{align*}
On the other hand, Lemma 2.2 of \citep{Haerdle88} ensures that
\begin{align*}
	\sup_{(u, v) \in [\delta, 1 - \delta]^2} \bigl\vert \overline h(u | v) - \E\bigl\{\overline h(u | v)\bigr\} \bigr\vert = O_{a.s.}\bigl(\sqrt{\ln n / (nb_n)}\bigr).
\end{align*}
Combining the previous two equations concludes the proof.
\end{proof} \noindent

 \begin{Lemma} \label{lem:h2}
Under conditions \ref{K1}, \ref{K2}, \ref{C1}, and \ref{C2}  it holds for all $(u,v) \in (0,1)^2$, $\delta \in (0, 0.5]$,  
	\begin{align*}
		\sup_{(u, v) \in [\delta, 1- \delta]^2} \bigl\vert \wh h(u|v) - \overline h(u|v) \bigl\vert  =  O_{a.s.}\bigl(a_{n}\bigr).
	\end{align*}
\end{Lemma}
\begin{proof}
	With arguments similar to the proof of \autoref{lem:c}, we can show
    \begin{align*}
        		&\sup_{(u, v) \in [\delta, 1- \delta]^2} \bigl\vert \wh h(u|v) - \overline h(u|v) \bigl\vert \\
		&\sup_{\bm z \in [\Phi^{-1}(\delta), \Phi^{-1}(1- \delta)]^2} \bigl\vert \wh h\bigl\{\Phi(z_1)| \Phi(z_2) \bigr\} - \overline h\bigl\{\Phi(z_1)| \Phi(z_2) \bigr\} \bigr\vert \\
        \le &\sup_{\bm z \in [\Phi^{-1}(\delta), \Phi^{-1}(1- \delta)]^2} \biggl\vert \frac{\eta_n(\bm z)}{\phi(z_2)} \times  \nabla_{\bm z} h\bigl\{\Phi(z_1)| \Phi(z_2) \bigr\} \biggr\vert \times O_{a.s.}(a_n), 
	\end{align*}
    where $\eta_n(\bm z) = \sup_{y \in [\min\{z_1,z_2\} - b_n, \max\{z_1, z_2\} + b_n]} 1 / \phi(y)$ and the $O_{a.s}$ term is independent of $\bm z$.  The supremum on the right hand side is $O(1)$ because all functions are continuous in $\bm z$ on every compact subset of $\R^2$. As a result, the right can be bounded by a constant times the $O_{a.s.}(a_n)$ term. This establishes our claim.
\end{proof}


\section{Proof of Proposition 5} \label{Appendix3}

From \autoref{prop:fhat}  and condition (v) in \autoref{prop:an} we get for all $\ell = 1, \dots, d,$ and $x \in \R$, that  $\wh f_\ell(x)  = f_\ell(x) + o_p\{b_{n,c}^2 + (nb_{n,c}^2)^{-1/2}\}$. This implies $(nb_{n,c}^2)^{1/2}\bigl\{\wh f_\ell(x) - f_\ell(x)\bigr\} = o_p(1)$ and we have established that the first $d$ components of \eqref{joint_an} converge to zero in probability. Hence, the first $d$ components of ${\bm \mu}_{\bm x}$ as well as the first $d$ rows and columns of $\Sigma_{\bm x}$ will be zero and we only have to deal with the remaining components in \eqref{joint_an}.

From \eqref{Theory:cproof_eq} and \eqref{Theory:Fproof_eq} in the proof of \autoref{Theory:rate_thm} and \autoref{prop:chat} we furthermore know that $\wh c_{j_e,k_e;D_e}(u,v) = \overline c_{j_e,k_e;D_e}(u,v) + o_p\{b_{n,c}^2 + (nb_{n,c}^2)^{-1/2}\}$ as well as $\wh F_{j_e|D_e}(x_{j_e}|\bm x_{D_e}) = F_{j_e|D_e}(x_{j_e}|\bm x_{D_e}) + o_p\{b_{n,c}^2 + (nb_{n,c}^2)^{-1/2}\}$. Similar to \autoref{lem:h2}, we can now show that
\begin{align*}
	&\overline c_{j_e,k_e;D_e}\bigl\{\wh F_{j_e|D_e}(x_{j_e}|\bm x_{D_e}), \wh F_{k_e|D_e}(x_{k_e}|\bm x_{D_e})\bigr\}\\
    =\;& \overline c_{j_e,k_e;D_e}\bigl\{F_{j_e|D_e}(x_{j_e}|\bm x_{D_e}), F_{k_e|D_e}(x_{k_e}|\bm x_{D_e})\bigr\} + o_p\{b_{n,c}^2 + (nb_{n,c}^2)^{-1/2}\}.
\end{align*}
Hence, for \eqref{joint_an} to hold it suffices to show that 
    \begin{align} \label{joint_an_c}
		(nb_{n, c}^2)^{1/2}\bigl\{\overline{\bm c}^*(\bm x) - b_{n,c}^2  \tilde{\bm\mu}_{\bm x} - {\bm c}^*(\bm x)\bigr\} \stackrel{d}{\to} \mathcal{N}\bigl(0, \tilde \Sigma_{\bm x}\bigr),
	\end{align}
     where
     
    $$\overline{\bm c}^*(\bm x) = \bigl(\overline c_{j_e,k_e;D_e}\{F_{j_e|D_e}(x_{j_e}|\bm x_{D_e}), F_{k_e|D_e}(x_{k_e}|\bm x_{D_e})\}\bigr)_{e \in E_1, \dots, E_{d-1}},$$ 
    and ${\bm c}^*(\bm x)$ is defined similarly, but replacing $\overline c_{j_e,k_e;D_e}$ with $c_{j_e,k_e;D_e}$. 
    
    Define $Z_{j_e|D_e}^{(i)} := \Phi^{-1}(U_{j_e|D_e}^{(i)} )$, $Z_{k_e|D_e}^{(i)} := \Phi^{-1}(U_{k_e|D_e}^{(i)})$,  $z_{j_e|D_e} := \Phi^{-1}\bigl\{F_{j_e|D_e}(x_{j_e}|\bm x_{D_e})\bigr\}$, $z_{k_e|D_e} := \Phi^{-1}\bigl\{F_{k_e|D_e}(x_{k_e}|\bm x_{D_e})\bigr\}$. Let $\bm Y_{n,i} := (Y_{n,i,e})_{e \in E_1, \dots, E_{d-1}}$, be a vector with entries
    \begin{align*}
		Y_{n,i,e} :=(nb_{n,c}^2)^{-1/2} \frac{K\biggl(\frac{Z_{j_e|D_e}^{(i)} - z_{j_e|D_e}}{b_n}\biggr) K\biggl(\frac{Z_{k_e|D_e}^{(i)} - z_{k_e|D_e}}{b_n}\biggr)}{ \phi(z_{j_e|D_e})\phi(z_{k_e|D_e})}.
	\end{align*}
    Then, $\sum_{i=1}^n \bm Y_{n,i} = (nb_{n,c}^2)^{1/2} \overline{\bm c}^*(\bm x)$.
    By the multivariate Lindeberg-Feller central limit theorem (Proposition 2.27 in \citep{vanderVaart98}), \eqref{joint_an_c} holds when  
    \begin{align}
	&\sum_{i=1}^n \E\bigl(\bm Y_{n, i}\bigr) = (nb_{n,c}^2)^{1/2}\bigl\{\bm c^*(\bm x) + b_{n,c}^2\tilde{\bm \mu}_{\bm x} + o(b_{n,c}^2)\bigr\}, \label{LF1} \\
    &\sum_{i=1}^n \cov(\bm Y_{n,i}) \to \tilde \Sigma_{\bm x}, \label{LF2} \\
    &\sum_{i=1}^n \E\bigl\{\| \bm Y_{n,i} \|^2 \ind\bigl(\| \bm Y_{n,i} \| > \varepsilon\bigr) \bigr\} \to 0, \quad \mbox{for all } \varepsilon > 0. \label{LF3}
	\end{align}
    Since $\bm Y_{n,i}$ are independent for $i = 1, \dots, n$, it holds 
   \begin{align*}
	\sum_{i=1}^n \E\bigl(\bm Y_{n, i}\bigr) = n \E\bigl(\bm Y_{n, i}\bigr), \qquad \sum_{i=1}^n \cov(\bm Y_{n,i}) = n \cov(\bm Y_{n,i}). 
\end{align*}
	Denote further $u_{j_e|D_e} := F_{j_e|D_e}(x_{j_e}|\bm x_{D_e})$,  $u_{k_e|D_e} := F_{k_e|D_e}(x_{k_e}|\bm x_{D_e})$. Corollary 3.4 in \citep{Nagler14} gives
    \begin{align} \label{bias_calc}
	n\E\bigl(Y_{n,i,e}\bigr) = (nb_{n,c}^2)^{1/2} \bigl\{c_{j_e,k_e;D_e}(u_{j_e|D_e}, u_{k_e|D_e}) + b_{n,c}^2 \tilde{\mu}_{\bm x, e} + o(b_{n,c}^2)\bigr\},
\end{align}
where 
\begin{align} 
	\tilde{\mu}_{\bm x, e} &:= \biggl\{\frac{\partial^2c_{j_e,k_e;D_e}\bigl(u_{j_e|D_e},u_{k_e|D_e}\bigr)}{\partial u_{j_e|D_e}^2}  \phi^2(z_{j_e|D_e}) + \frac{\partial^2 c_{j_e,k_e;D_e}\bigl(u_{j_e|D_e}, u_{k_e|D_e}\bigr)}{\partial u_{k_e|D_e}^2}  \phi^2(z_{k_e|D_e}) \notag \\
       \begin{split}
	&\phantom{:=\biggl[}- \frac{3 \partial c_{j_e,k_e;D_e}\bigl(u_{j_e|D_e},u_{k_e|D_e}\bigr)}{\partial u_{j_e|D_e}}  \phi(z_{j_e|D_e})z_{j_e|D_e} \\ 
    &\phantom{:=\biggl[} - \frac{3 \partial c_{j_e,k_e;D_e}\bigl(u_{j_e|D_e},u_{k_e|D_e}\bigr)}{\partial u_{k_e|D_e}} \phi(z_{k_e|D_e})z_{k_e|D_e} 
    \end{split} \label{el_mu} \\
	&\phantom{:=\biggl[}  + c_{j_e,k_e;D_e}\bigl(u_{j_e|D_e}, u_{k_e|D_e}\bigr) \times  \bigl(z_{j_e|D_e}^2 + z_{k_e|D_e}^2-2\bigr)  \biggr\} \frac {\sigma_K^2}{2}, \notag 
\end{align}
and $\sigma_K^2 := \int_{[-1, 1]}x^2K(x)dx$. This validates  \eqref{LF1}. By the change of variable $s_1 = (z_1 - z_{j_e|D_e})/b_{n,c}$, $s_2 = (z_2 - z_{k_e|D_e})/b_{n,c}$, and a Taylor approximation of $\psi_{j_e, k_e;D_e}$ (as defined in \ref{C2}), we get
\begin{align}
&  n\E\bigl(Y_{n,i,e}^2\bigr) \phi^2(z_{j_e|D_e})\phi^2(z_{k_e|D_e}) \notag \\
=\; & n\E\biggl\{\frac{1}{nb_{n,c}^2} K^2\biggl(\frac{Z_{j_e|D_e}^{(i)} - z_{j_e|D_e}}{b_{n,c}}\biggr) K^2\biggl(\frac{Z_{k_e|D_e}^{(i)} - z_{k_e|D_e}}{b_{n,c}}\biggr)  \biggr\} \notag  \\
    =\; & \int_\R \int_\R K^2(s_1)K^2(s_2)\psi_{j_e, k_e;D_e}(z_{j_e|D_e} - b_{n,c} s_1, z_{k_e|D_e} - b_{n,c} s_2) ds_1 ds_2 \notag \\
    = \;& \nu_K^2 \psi_{j_e, k_e;D_e}(z_{j_e|D_e}, z_{k_e|D_e}) + o(1), \label{var_calc}
\end{align}
where $\nu_{K} := \int_\R K^2(s)ds$. Using \eqref{var_calc} and \eqref{bias_calc}, we obtain
\begin{align} \label{el_sig}
	n\var(Y_{n, i, e})  
    \to \nu_K^2 \frac{c_{j_e,k_e;D_e}\bigl(u_{j_e|D_e},u_{k_e|D_e}\bigr)}{\phi(z_{j_e|D_e})\phi(z_{k_e|D_e})} =: \tilde \sigma_{\bm x, e}.
\end{align}
Arguments similar to \eqref{var_calc} show that for any two edges $e \neq e^\prime$, it holds $n\E(Y_{n,i,e}Y_{n,i,e^\prime}) = O(b_{n,c})$; and with \eqref{bias_calc}, $n \cov(Y_{n,i,e}, Y_{n,i,e^\prime}) \to 0$. We have shown that \eqref{LF2} holds with $\tilde \Sigma_{\bm x}$ being a diagonal matrix with diagonal entries $\tilde \sigma_{\bm x, e}$ given in \eqref{el_sig}.

Instead of checking the remaining condition \eqref{LF3} directly, we will verify the stronger Lyapunov-type condition $\sum_{i=1}^n \E(\| \bm Y_{n,i} \|^3 ) \to 0$. By Jensen's inequality we get 
    \begin{align*}
    n\E\bigl(\| \bm Y_{n,1} \|^3 \bigr) = nE\biggl\{\biggl(\sum_{m = 1}^{d - 1}\sum_{e \in E_m}  Y_{n,1,e}^2 \biggr)^{3/2}\biggr\} 
    &\le n\sqrt{d(d-1)/2} \sum_{m = 1}^{d - 1}\sum_{e \in E_m} \E\bigl( Y_{n,1,k}^3 \bigr),
\end{align*}
where $d(d-1)/2$ is the number of terms in the double sum. Hence, it suffices to show $n\E( Y_{n,1,e}^3) \to 0$ for any $e \in E_1, \dots, E_{d-1}$. Similar to \eqref{var_calc}, we get $n\E( Y_{n,1,k}^3) =  O\{1/(nb_n^2)^{1/2}\}$ which is $o(1)$. $\qed$

	\bibliography{References}
		
\end{document}